\documentclass[10pt,twocolumn,journal]{IEEEtran}

\usepackage{amsmath, mathrsfs}

\makeatletter
\def\maketag@@@#1{\hbox{\m@th\normalfont\normalsize#1}}
\makeatother

\newcommand{\subparagraph}{}
\usepackage[compact]{titlesec}
\titlespacing*{\section}{2pt}{1\baselineskip}{0.9\baselineskip}

\allowdisplaybreaks

\usepackage{amsfonts}
\usepackage{amssymb}
\usepackage{color}
\usepackage{multirow}
\usepackage{graphicx}

\usepackage{caption}
\usepackage{subcaption}
\usepackage[numbers,sort&compress]{natbib}
\usepackage{balance}

\def\mindex#1{\index{#1}}

\def\sq{\hbox{\rlap{$\sqcap$}$\sqcup$}}
\def\qed{\ifmmode\sq\else{\unskip\nobreak\hfil
\penalty50\hskip1em\null\nobreak\hfil\sq
\parfillskip=0pt\finalhyphendemerits=0\endgraf}\fi\medskip}

\long\def\defbox#1{\framebox[.9\hsize][c]{\parbox{.85\hsize}{%
\parindent=0pt
\baselineskip=12pt plus .1pt      
\parskip=6pt plus 1.5pt minus 1pt 
 #1}}}

\long\def\beginbox#1\endbox{\subsection*{}%
\hbox{\hspace{.05\hsize}\defbox{\medskip#1\bigskip}}%
\subsection*{}}

\def\endbox{}

\def\diag{{\text{diag}}}

\def\tr{\mathsf{tr}}

\newsavebox{\junk}
\savebox{\junk}[1.6mm]{\hbox{$|\!|\!|$}}

\def\argmin{\mathop{\rm arg\, min}}

\def\Re{\field{R}}

\def\bC{{\mathbb C}}

\def\bE{{\mathbb E}}

\def\bR{{\mathbb R}}
\def\bS{{\mathbb S}}

\def\bfA{{\bf A}}
\def\bfB{{\bf B}}
\def\bfC{{\bf C}}
\def\bfD{{\bf D}}

\def\bfG{{\bf G}}

\def\bfI{{\bf I}}

\def\bfK{{\bf K}}

\def\bfP{{\bf P}}

\def\bfR{{\bf R}}
\def\bfS{{\bf S}}

\def\bfU{{\bf U}}
\def\bfV{{\bf V}}
\def\bfW{{\bf W}}
\def\bfX{{\bf X}}
\def\bfY{{\bf Y}}
\def\bfZ{{\bf Z}}

\def\bfa{{\bf a}}

\def\bfg{{\bf g}}
\def\bfh{{\bf h}}

\def\bfm{{\bf m}}
\def\bfn{{\bf n}}

\def\bfp{{\bf p}}
\def\bfq{{\bf q}}
\def\bfr{{\bf r}}

\def\bfu{{\bf u}}
\def\bfv{{\bf v}}
\def\bfw{{\bf w}}
\def\bfx{{\bf x}}
\def\bfy{{\bf y}}

\def\scrI{{\mathscr{I}}}

\def\scrU{{\mathscr{U}}}
\def\scrV{{\mathscr{V}}}

\def\sfH{{\sf H}}

\def\bfmath#1{{\mathchoice{\mbox{\boldmath$#1$}}%
{\mbox{\boldmath$#1$}}%
{\mbox{\boldmath$\scriptstyle#1$}}%
{\mbox{\boldmath$\scriptscriptstyle#1$}}}}

\def\bfmY{\bfmath{Y}}

\def\bfmhhaY{\bfmath{\hhaY}} 
\def\bfmhhaY{\hbox to 0pt{$\widehat{\bfmY}$\hss}\widehat{\phantom{\raise 1.25pt\hbox{$\bfmY$}}}}

\def\til={{\widetilde =}}

\def\clB{{\cal B}}
\def\clC{{\cal C}}
\def\clD{{\cal D}}

\def\clG{{\cal G}}

\def\clI{{\cal I}}

\def\clN{{\cal N}}

\def\clP{{\cal P}}

\def\clR{{\cal R}}
\def\clS{{\cal S}}
\def\clT{{\cal T}}

\def\clW{{\cal W}}
\def\clX{{\cal X}}

 \def\FRAC#1#2#3{\genfrac{}{}{}{#1}{#2}{#3}}

\def\ddtp{{\mathchoice{\FRAC{1}{d^{\hbox to 2pt{\rm\tiny +\hss}}}{dt}}%
{\FRAC{1}{d^{\hbox to 2pt{\rm\tiny +\hss}}}{dt}}%
{\FRAC{3}{d^{\hbox to 2pt{\rm\tiny +\hss}}}{dt}}%
{\FRAC{3}{d^{\hbox to 2pt{\rm\tiny +\hss}}}{dt}}}}

\def\average#1,#2,{{1\over #2} \sum_{#1}^{#2}}

\def\eye(#1){{\bf(#1)}\quad}

\newtheorem{definition}{{\bf Definition}}
\newtheorem{proposition}{{\bf Proposition}}

\newtheorem{remark}{{\bf Remark}}

\def\eq#1/{(\ref{e:#1})}

\newcommand{\inp}[2]{{\langle #1, #2 \rangle}}
\newcommand{\inpr}[2]{{\langle #1, #2 \rangle}_\bR}

\newcommand{\beqn}[1]{\notes{#1}%
\begin{eqnarray} \elabel{#1}}

\newcommand{\eeqn}{\end{eqnarray} }

\newcommand{\beq}[1]{\notes{#1}%
\begin{equation}\elabel{#1}}

\newcommand{\eeq}{\end{equation}}

\def\bdes{\begin{description}}
\def\edes{\end{description}}

\newcounter{rmnum}

\newcounter{anum}

{\end{list}}

\def\ass(#1:#2){(#1\ref{#1:#2})}

\def\ritem#1{
\item[{\sf \ass(\current_model:#1)}]
}

\newenvironment{recall-ass}[1]{%
\begin{description}
\def\current_model{#1}}{
\end{description}
}

\long\def\comment#1{}

\newfont{\bbb}{msbm10 scaled 700}

\newfont{\bb}{msbm10 scaled 1100}

\newcommand{\xv}{{\bf x}}

\newcommand{\Id}{{\bf I}}

\newcommand{\Mm}{{\bf M}}

\newcommand{\Xm}{{\bf X}}

\newcommand{\Gammam}{\boldsymbol{\Gamma}}
\newcommand{\Lambdam}{\boldsymbol{\Lambda}}
\newcommand{\Deltam}{\boldsymbol{\Delta}}
\newcommand{\Sigmam}{\boldsymbol{\Sigma}}

\renewcommand{\Re}{{\rm Re}}

\newcommand{\transp}{{\sf T}}

\usepackage{tikz}
\usepackage{pgfplots}
\pgfplotsset{compat=newest}
\usetikzlibrary{patterns}
\usetikzlibrary{positioning}
\usetikzlibrary{datavisualization}
\usetikzlibrary{datavisualization.formats.functions}
\usetikzlibrary{backgrounds}
\usetikzlibrary{shapes,snakes}
\usepackage{textcomp}

\usepackage[none]{hyphenat}

\usepackage{algorithm}
\usepackage{algpseudocode}
\usepackage{pifont}
\usepackage{varwidth}

\def\herm{{\sfH}}

\def\prox{{\mathsf{prox}}}
\def\cg{{\clC\clN}} 
\def\matlab{{MATLAB\textcopyright\,}}

\begin{document}

\title{Low-Complexity Massive MIMO Subspace Estimation and Tracking from Low-Dimensional Projections}
\author{Saeid Haghighatshoar,  \IEEEmembership{Member, IEEE,} Giuseppe Caire,
\IEEEmembership{Fellow, IEEE}%
\thanks{The authors are with the Communications and Information Theory Group, Technische Universit\"{a}t Berlin (\{saeid.haghighatshoar, caire\}@tu-berlin.de).}
\thanks{A shorter version of this paper was presented in the IEEE International Conference on Communications (ICC), Paris, France, May 2017.}

}

\maketitle

\begin{abstract}
Massive MIMO is a variant of multiuser MIMO,
in which the number of antennas $M$ at the base-station is very
large, and generally much larger than the number of spatially
multiplexed data streams to/from the users. It has been observed that in many realistic propagation scenarios  as well as in spatially correlated channel models used in standardizations, although the user channel vectors have a very high-dim $M$, they lie on low-dim subspaces due to their limited angular spread (spatial correlation). 
This low-dim subspace structure remains stable across many coherence blocks and can be exploited in several ways to improve the system performance. A main challenge, however, is to estimate this signal subspace from samples of users' channel vectors as fast and efficiently as possible.
In a recent work, we addressed this problem and proposed a very effective novel algorithm referred to as Approximate Maximum-Likelihood (AML), which was formulated as a semi-definite program (SDP). In this paper, we address two problems left open in our previous work, namely, computational complexity and tracking. 
The algorithm proposed in this paper is reminiscent of Multiple Measurement Vectors (MMV) problem in Compressed Sensing and is proved to be equivalent to the AML Algorithm for sufficiently dense angular grids. 
It has also a very low computational complexity and is able to track sharp transitions in the channel statistics very quickly. Although mainly motivated by massive MIMO applications, our proposed algorithm is of independent interest in other related subspace estimation applications. We assess the estimation/tracking performance of our proposed algorithm empirically via  numerical simulations, especially in  practically relevant situations where a direct implementation of the SDP would be infeasible in real-time. We also compare the performance of our algorithm with other related subspace estimation/tracking algorithms in the literature. 
\end{abstract}

\section{Introduction}  \label{sec:intro}

\PARstart{C}{onsider} a multiuser massive MIMO system formed by a base-station (BS) with $M$ antennas serving $K$ single-antenna mobile 
users in a cellular system. 
Following the current {\em massive MIMO} approach 
\cite{Marzetta-TWC10,Huh11,hoydis2013massive,larsson2014massive}, we focus on 
uplink (UL) and downlink (DL) in Time Division Duplexing (TDD), where the base-station (BS) transmit/receive hardware 
is designed or calibrated in order to preserve UL-DL reciprocity \cite{shepard2012argos,rogalin2014scalable} 
such that the BS estimates the channel vectors of the users from UL orthogonal training pilots 
sent by the users and uses them to transmit data to the users in the DL via coherent beamforming.  Since there is no multiuser interference in the UL training phase (after neglecting the pilot contamination), 
in this paper we focus on the basic channel estimation problem for a single user. 

In massive MIMO systems, the number of antennas $M$, thus, the dimension of the received signal at the BS is very large. However, in many relevant scenarios, channel vectors of each user are spatially  correlated since the propagation between the user and the BS occurs through a small 
set of \textit{Angle of Arrivals} (AoAs). Denoting by $\bfh(t) \in \bC^M$ the channel vector of a generic user at a time slot $t$, this implies that the signal covariance matrix of $\bfh(t)$ given by $\bfS:=\bE[\bfh(t) \bfh(t)^\herm]$ is typically low-rank. 
This  low-rank structure can be exploited to improve the system multiplexing gain and decrease 
the training overhead.  A particularly effective scheme is the Joint Spatial Division and Multiplexing (JSDM) approach proposed and analyzed in
\cite{adhikary2013joint,nam2014joint,adhikary2014joint}, where the users are partitioned into $G > 1$ groups such that users in each group have similar channel subspaces 
\cite{adhikary2013joint,nam2014joint,adhikary2014joint}. These groups are separated by a zero-forcing beamforming that uses only the group subspace information and reduces the dimensionality for each group $g$ to some $m_g \ll M$.
Then, additional multiuser multiplexing gain in each group $g$ is obtained by applying the conventional linear precoding to the lower-dim projected channel. 
This has the additional non-trivial advantage that only $m \ll M$ RF chains (A/D converters and modulators) are needed, thus, reducing the A/D conversion rate significantly.
This and many other related examples evidently illustrate that estimation of the signal subspace of the users plays a crucial rule in massive MIMO systems. However, obtaining the signal subspace information in massive MIMO is a high-dim estimation problem and becomes quite challenging as the number of antennas $M$ increases, especially that, due to the limited number of available RF chains at the receiver front end, the subspace estimation needs to be done with only low-dim projections of the received signal.

\subsection{Contribution}
In our recent work \cite{haghighatshoar2016channel, haghighatshoar2016massive}, we studied this problem and developed a new family of efficient algorithms for subspace estimation in massive MIMO. We also demonstrated via numerical simulations that our proposed algorithms provide near-ideal performance for a massive MIMO JSDM system. However, the low-complexity implementation of our proposed algorithms was left open in \cite{haghighatshoar2016channel, haghighatshoar2016massive}.
In this paper, we bridge the complexity gap by providing efficient and low-complexity implementation of the algorithms in our previous work \cite{haghighatshoar2016channel, haghighatshoar2016massive}, with a special focus on the AML (approximate maximum likelihood) Algorithm. Our approach is based on approximating the typically high-complexity semi-definite program (SDP) proposed for the original form of AML Algorithm in \cite{haghighatshoar2016channel, haghighatshoar2016massive} with another convex optimization problem that can be efficiently solved. We consider a generalization of the originally proposed AML Algorithm where the projection (sampling) operator may be time-variant, i.e., changing in different training slots. This results in further improvement in the subspace estimation. We extend our  proposed low-complexity algorithm to more practical array configurations such as 2D rectangular lattice arrays, and provide guidelines for efficient numerical implementation for general array configurations. We also illustrate that our algorithm can be run in a tracking mode, where the subspace estimate is updated upon arrival of a new training sample.

\subsection{Related Work}
Subspace estimation and tracking arises in a variety of problems in signal processing. 
In general, whenever a high-dim signal is generated by a linear process that is governed by a small number of parameters, it can be represented by a low-dim structure embedded in a  higher-dim space. This occurs in a wide range of applications such as  \textit{Direction-of-Arrival} (DoA) estimation \cite{kumaresan1983estimating, roy1989esprit, schmidt1986multiple}, source localization \cite{shahbazpanahi2001distributed}, anomaly detection \cite{stein2002anomaly}, adaptive filtering \cite{sayed2003fundamentals}, and wireless communication \cite{tse2005fundamentals}. The dominant signal subspace is obtained by computing the covariance matrix of the data (i.e., its second order statistics), e.g., the channel covariance matrix $\bfS=\bE[\bfh(t) \bfh(t)^\herm]$ in massive MIMO applications, and calculating its \textit{Singular Value Decomposition} (SVD), and is classically known as \textit{Principal Component Analysis} (PCA) in statistics  \cite{jolliffe1986principal} and \textit{Karhunen-Lo\`eve Transform} (KLT) in stochastic processes \cite{papoulis2002probability}. 
In practice, however, the data covariance matrix is not a priori known and should be estimated from the observed data samples. In particular, due to complexity reasons, this needs to be done by taking as few data samples and by consuming as less storage as possible. This has motivated a vast line of research on developing efficient and low-complexity subspace estimation/tracking algorithms \cite{comon1990tracking, yang1995projection, crammer2006online, balzano2010online, chi2013petrels}. 
Recently by the advent of Compressed Sensing \cite{donoho2006compressed, candes2006near}, there has been a new serge of interest in exploiting \textit{low-dim signal structures} such as sparsity and low-rankness, which has revitalized the popularity of subspace techniques in a variety of problems including Matrix Completion \cite{candes2009exact, candes2010power}, Super-Resolution \cite{candes2014towards}, compressive spectral estimation  \cite{baraniuk2007compressive, herman2009high}, and sparse signal reconstruction \cite{bajwa2010compressed, tropp2006algorithms, lee2012subspace, tropp2006algorithms2, malioutov2005sparse}

We compare the performance of our proposed algorithm with the \textit{Singular Value Thresholding} (SVT) Algorithm \cite{cai2010singular} in the \textit{batch mode} and with PETRELS Algorithm \cite{chi2013petrels} in the \textit{online tracking mode}. Both algorithms provide the state of the art performance in subspace estimation (SVT) and tracking (PETRELS). Our numerical simulations in Section \ref{sec:sim} illustrate that our algorithm performs better that SVT in the batch mode and is able to track  sharp transitions in signal statistics much faster than PETRELS in the online mode.

\subsection{Notation}
We show vectors by boldface small letters (e.g., $\xv$), matrices by boldface capital letters (e.g., $\Xm$), scalar constant by 
non-boldface letters (e.g., $x$ or $X$), and sets by calligraphic letters (e.g., $\clX$).
The $i$-th element of a vector $\xv$
and the $(i,j)$-th element of a matrix $\Xm$ is denoted by $[\xv]_i$ and $[\Xm]_{i,j}$. 
We represent the $i$-th row and $j$-th column of a matrix $\bfX$ with a row vector $\Xm_{i,.}$ and a column vector $\Xm_{.,j}$.
We denote the Hermitian and the transpose of a matrix (or a vector) $\bfX$ by $\bfX^\herm$ and $\Xm^\transp$.
We use $\tr(.)$ for the trace operator.
We denote the complex/real inner product of two matrices (or two vectors) $\bfX$ and $\bfY$ by $\inp{\bfX}{\bfY}=\tr(\bfX^\herm \bfY)$, and $\inpr{\bfX}{\bfY}=\Re[\inp{\bfX}{\bfY}]$. We use $\|\bfx\|$ for the $l_2$-norm of a vector $\bfx$, and $\|\bfX\|=\inp{\bfX}{\bfX}=\inpr{\bfX}{\bfX}$ for the Frobenius norm of a matrix $\bfX$.
We denote a $k\times k$ diagonal matrix with diagonal elements $s_1, \dots, s_k$ with $\diag(s_1, \dots, s_k)$.
We indicate the output of any optimization problem such as $\argmin_{\bfX \in \clX} f(\bfX)$ with $\bfX^*$.
The identity matrix of order $p$ is denoted by $\Id_p$. For an integer $k$, we use the shorthand notation $[k]$ for $\{1, \dots, k\}$. We denote the big-O notation by $O(.)$.

\section{Basic Setup}
\subsection{Array and Signal Model}
 Consider a BS with a large \textit{Uniform Linear Array} (ULA) with $M \gg 1$ antennas. 
The geometry of the array is shown in Fig.\,\ref{fig:sc_channel}, with antenna elements having a uniform spacing $d$. 
\begin{figure}[t]
\centering
\includegraphics{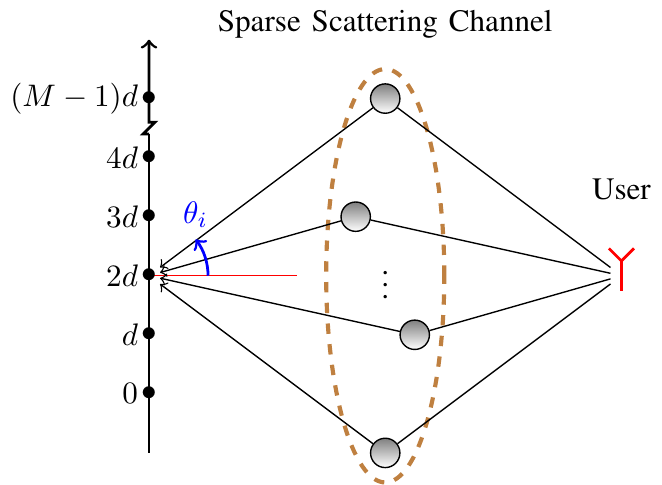}
\caption{{\small 
Array configuration in a multi-antenna receiver in the presence of a sparse scattering channel with only few scatterers with discrete angle of arrivals.}}
\label{fig:sc_channel}
\end{figure}
We assume that the BS serves the users in the angular range $[-\theta_{\max}, \theta_{\max}]$ for some $\theta_{\max} \in (0,\pi/2)$, and  set 
$d=\frac{\lambda}{2 \sin(\theta_{\max})}$, where $\lambda=\frac{c_0}{f_c}$ denotes the wave-length,  where $f_c$ is the carrier frequency and $c_0$ is speed of the light.
We consider a simple propagation model  
in which the transmission between a user and the BS occurs through 
$p$ scatterers (see Fig.\,\ref{fig:sc_channel}). The results can be simply extended to a general scattering model with a general mixed-type (continuous and discrete) power distribution over the AoA domain as in  \cite{haghighatshoar2016massive}. One snapshot of the received signal in a window of training pilots of size $T$ is given by
\begin{equation} \label{ziofafa}
\bfy(t)=\bfh(t) z(t)+\bfn(t):=\sum_{l=1}^p \bfa(\theta_l) w_l(t)\,  z(t) + \bfn(t),
\end{equation}
where $\bfh(t):=\sum_{l=1}^p \bfa(\theta_l) w_l(t)\in \bC^M$ denotes the channel vector of the user\footnote{In this paper, $\bfh(t)$ denotes the channel vector of a generic user at a specific subcarrier over an OFDM symbol transmitted at a time slot $t\in [T]$ (a specific time-frequency tile). Although the instantaneous channel vectors, which are used to transmit/receive data to/from the users in the DL/UL via coherent beamforming, might be highly correlated, we always assume that the channel vectors used for subspace estimation are sufficiently separated in time such that $\bfh(t)$, $t \in [T]$ are independent, where they are also identically distributed (i.i.d.) due to channel stationarity in time. Even in a fixed time slot $t\in[T]$,  one can obtain i.i.d. realizations of the channel vector by sampling at sufficiently separated subcarriers in the frequency domain since the channel  is also stationary in frequency with a frequency-invariant second order statistics. Thus, having multiple subcarriers in the frequency domain has the same effect as taking more observations in the time domain.}, $z(t)\in \bC$ is the transmitted pilot  symbol of the user, which typically belongs to a signal constellation such as QAM, $w_l(t)\sim \cg(0, \sigma_l^2)$ is the channel gain 
of the $l$-th multipath component, $\bfn(t) \sim \cg({\bf 0}, \sigma^2\Id_M)$ is the additive white Gaussian noise of the  antenna elements, and where $\bfa(\theta) \in \bC^M$ is the array response at AoA $\theta$, whose $k$-th component is given by 
\begin{align}
[\bfa(\theta)]_k
= e^{j k\frac{2\pi d \sin(\theta)}{\lambda}}
=e^{j k\pi \frac{\sin(\theta)}{\sin(\theta_{\max})}}.
\end{align}
According to the well-known and widely-accepted \textit{Wide Sense Stationary Uncorrelated Scattering} (WSSUS) model, the channel gains of different paths, i.e., $\{w_l(t)\}_{l=1}^p$, at every time $t \in [T]$,
are uncorrelated \cite{tse2005fundamentals}. 
Without loss of generality, we assume that the pilot symbol $z(t)$ transmitted along the channel vector $\bfh(t)$ is normalized to $z(t)=1$ in all training snapshots, thus, letting 
$\bfA = [\bfa(\theta_1), \dots,  \bfa(\theta_p)]$, we have 
\begin{align}\label{eq:disc_ch_mod}
\bfy(t)=\bfh(t)+\bfn(t)=\bfA \bfw(t) + \bfn(t), \;\;\; t\in[T],
\end{align}
where $\bfw(t) = (w_1(t), \dots, w_p(t))^\transp$ for different $t\in [T]$ are statistically independent. 
Also,  we assume that the AoAs $\{\theta_l\}_{l=1}^p$ remain invariant over the whole training period of length $T$ slots. From \eqref{eq:disc_ch_mod}, the covariance of $\bfy(t)$ is given by
\begin{align}\label{eq:subs_embed}
{\bfC}_y=\bfA \mathbf{\Sigma} \bfA^\herm + \sigma^2\Id_M= \sum_{l=1}^p \sigma_l^2 \bfa(\theta_l) \bfa(\theta_l)^\herm + \sigma^2\Id_M,
\end{align}
where $\Sigmam=\diag(\sigma_1^2, \dots, \sigma_p^2)$ is the covariance matrix of $\bfw(t)$, and where $\bfS=\bE[\bfh(t) \bfh(t)^\herm]=\sum_{l=1}^p \sigma_l^2 \bfa(\theta_l) \bfa(\theta_l)^\herm$ is the covariance matrix of the channel vectors. It is not difficult to check that $\bfS$ is a Hermitian \textit{positive semi-definite} (PSD) Toeplitz matrix of rank $p$, where typically $p \ll M$. In practice, the channel vectors are formed by the superposition of AoA contributions weighted according to a mixed-type measure $\gamma(d\theta)$ as in \cite{haghighatshoar2016massive}, containing both discrete masses in correspondence of specular reflectors and a continuous measure corresponding to scattering clusters. In this case, $\bfS$ is given by $\int \gamma(d\theta) \bfa(\theta)^\herm \bfa(\theta)$ and is generally full rank (algebraic rank). However, since $\gamma$ has a limited angular support in practice, $\bfS$ can be well approximated by a  low-rank matrix (effective rank), thus, the low-rank assumption is still valid. The AML Algorithm in \cite{haghighatshoar2016channel, haghighatshoar2016massive} and also our proposed low-complexity scheme in this paper apply to this general case.

\subsection{Sampling Operator}\label{sec:samp_op}
As explained in the introduction, in massive MIMO systems, it is crucial to be able to recover the signal subspace of the users from low-dim projections of their received channel vectors. 
In general, low-dimensional projections can be obtained via a  $m\times M$ matrix $\bfB$ for some $m\ll M$, which can be implemented as part of the analog receiver front-end. A particularly simple and attractive choice 
is ``antenna selection'', where $\bfB$ is a binary 0-1 selection matrix with a single element equal to $1$ in each row. 
In this paper, we always consider such an antenna selection scheme as the projection operator, where in each training slot, the BS samples the output signal of only $m \ll M$ random antenna elements via $m$ available RF chains. Also, we consider a general case in which the antenna selection can be time-variant. Letting $\clI_t=\{i_1(t), \dots,  i_m(t)\} \subseteq [M]$ be the indices of $m$ randomly selected antenna elements at time $t \in [T]$, we denote the $m\times M$ projection matrix by $\bfB(t)$, where the single $1$ in each row is given by $[\bfB(t)]_{k,i_k(t)}=1$, for $k\in [m]$. Note that for such a matrix $\bfB(t)$, we have that $\bfB(t) \bfB(t)^\herm=\bfI_m$. We define the noisy projection (sketch) at time $t\in [T]$ by 
$\bfx(t):=\bfB(t) \bfy(t)=\bfB(t) \bfh(t) + \widetilde{\bfn}(t)$, 
where $\bfy(t)$ is given by \eqref{eq:disc_ch_mod} and where $\widetilde{\bfn}(t):=\bfB(t) \bfn(t) \sim \cg(0, \sigma^2 \bfI_m)$ denotes the projected noise vector.

\subsection{Performance Metric}\label{sec:perf_metric}
Our goal is to find an estimate of the dominant signal subspace of the covariance matrix $\bfS$ of the channel vectors $\bfh(t)$, $t\in [T]$. Let $\widehat{\bfS}$ be such an estimate and let $\bfS=\bfU \Lambdam \bfU^\herm$ and $\widehat{\bfS}=\widehat{\bfU} \widehat{\Lambdam} \widehat{\bfU}^\herm$ denote the SVD of $\bfS$ and $\widehat{\bfS}$. We always use the convention that the singular values (e.g., those in $\Lambdam$ and $\widehat{\Lambdam}$) are sorted in a non-increasing order. We define the normalized power distribution for $\bfS$ by a vector $\bfp \in \bR_+^M$, where $[\bfp]_i=\frac{\lambda_i}{\sum_{j=1}^M \lambda_j}$ with  $\lambda_k$  denoting the $k$-th singular value of $\bfS$. Let $\widehat{\bfU}=[\widehat{\bfu}_1, \dots, \widehat{\bfu}_M]$ where $\widehat{\bfu}_i$ denotes the $i$-th column of $\widehat{\bfU}$. We denote the power captured by columns of $\widehat{\bfU}$ by a vector $\bfq \in \bR_+^M$, where $[\bfq]_i:=\inp{\bfS}{\widehat{\bfu}_i \widehat{\bfu}_i^\herm}=\bE[|\widehat{\bfu}_i^\herm \bfh(t)|^2]$ gives the amount of power of $\bfS$ captured by the $1$-dim (rank-1) projection operator $\widehat{\bfu}_i \widehat{\bfu}_i^\herm$. Note that $\sum_{i=1}^M [\bfq]_i= \inp{\bfS}{\widehat{\bfU} \widehat{\bfU}^\herm}=\tr(\bfS)=\bE[\|\bfh(t)\|^2]$, which gives the whole power contained in $\bfS$. We normalize $\bfq$ and define the estimated normalized power distribution $\widehat{\bfp}\in \bR_+^M$, where $[\widehat{\bfp}]_i=\frac{[\bfq]_i}{\sum_{j=1}^M [\bfq]_j}$. Let $\eta_\bfp(k):=\sum_{i=1}^k [\bfp]_i$ and $\eta_{\widehat{\bfp}}(k)=\sum_{i=1}^k [\widehat{\bfp}]_i$, for $k\in[M]$, denote the whole normalized signal power contained in the first $k$ component of $\bfp$ and $\widehat{\bfp}$. Note that since $\bfU$ is the SVD basis for $\bfS$, we always have $\eta_{\bfp}(k) \geq \eta_{\widehat{\bfp}}(k)$, for every $k \in [M]$, which implies that the vector $\widehat{\bfp}$ is always majorized by ${\bfp}$ \cite{marshall1979inequalities}. Also, due to the normalization, we have $\eta_{\bfp}(M)=\eta_{\widehat{\bfp}}(M)=1$.

In subspace estimation applications in massive MIMO, e.g., in JSDM, the goal is to design for each user a low-dim beamformer that captures a significant amount of the power of its channel vectors\footnote{Namely, an $M\times q$ matrix $\bfV$, for some $q\ll M$, satisfying $\bfV^\herm \bfV=\bfI_q$, and $\bE\big[\|\bfV^\herm \bfh(t)\|^2\big] \geq (1-\epsilon) \bE\big[\|\bfh(t)\|^2\big]$, for some small $\epsilon \in (0,1)$.}. An appropriate distortion measure for such applications is $\nu(\bfp, \widehat{\bfp}):=\max_{k \in [M]} \frac{\eta_{{\bfp}}(k) - \eta_{\widehat{\bfp}} (k)}{\eta_{{\bfp}}(k)}$, which captures the maximum ratio of power loss incurred by beamforming to the dominant $k$-dim subspace of the estimate $\widehat{\bfS}$ rather than the optimal $k$-dim subspace of $\bfS$, for any arbitrary $k \in [M]$. We will use $\Gamma(\bfp, \widehat{\bfp}):=1-\nu(\bfp, \widehat{\bfp})$ as the metric for assessing the performance of the subspace estimation. Note that $\Gamma(\bfp, \widehat{\bfp})\in [0,1]$, where  $\Gamma(\bfp, \widehat{\bfp})=1$ if and only if $\bfS=\mu \widehat{\bfS}$ for some $\mu >0$. In particular, if $\Gamma(\bfp, \widehat{\bfp})\geq 1-\epsilon$, for some fixed $\epsilon \in (0,1)$, then the $M\times k$ beamforming matrix $[\widehat{\bfu}_1, \dots, \widehat{\bfu}_k]$ obtained from the estimate $\widehat{\bfS}$ for an arbitrary $k \in [M]$ is at least $(1-\epsilon)$-optimal with respect to the best $k$-dim subspace of true signal covariance  matrix $\bfS$.

\section{Problem Statement}
In this section, we briefly explain the AML Algorithm proposed for subspace estimation in \cite{haghighatshoar2016channel, haghighatshoar2016massive}. For simplicity of explanation, we assume that the 0-1 sampling operator is a fixed $m\times M$ matrix $\bfB$ for all $t \in [T]$. We will later consider the generalized time-variant sampling operator $\bfB(t)$. 

Let $\bfy(t)=\bfh(t)+\bfn(t)$ be the noisy user channel vector received at the array at time $t \in [T]$, and let $\bfx(t)=\bfB \bfy(t)$ be its $m$-dim projection via $\bfB$. We assume that the noise variance $\sigma^2$ is known and multiply all the signals by $\frac{1}{\sigma}$ to normalize the noise power to $1$. For simplicity, we still use the same notation for the normalized signals. Let $\widehat{\bfC}_x=\frac{1}{T}\sum_{t \in [T]} \bfx(t) \bfx(t)^\herm$ be the sample covariance matrix of the sketches $\bfx(t)$, $t \in [T]$, let $\widehat{\bfC}_x=\bfV \bfD \bfV^\herm$ be its SVD, and define $\Deltam:=\widehat{\bfC}^{1/2}_x=\bfV \bfD^{1/2}$. The AML Algorithm in \cite{haghighatshoar2016channel, haghighatshoar2016massive} is cast as the following SDP:
\begin{align}\label{eq:aml}
({\bfS}^*,&\bfK^*) = \argmin_{\Mm \in \clT_+, \bfK\in \bC^{m\times m}} \tr(\bfB \Mm \bfB^\herm) + \tr(\bfK)\nonumber\\
& \text{ subject to } \left [  \begin{array}{cc}  \bfI_m+ \bfB \Mm \bfB^\herm & \mathbf{\Delta} \\ \mathbf{\Delta}^\herm & \bfK \end{array} \right ] \succeq {\bf 0},
\end{align}
where $\clT_+$ denotes the space of all $M\times M$ Hermitian PSD Toeplitz matrices. The optimal solution $\bfS^*$ of \eqref{eq:aml} gives an estimate of the covariance matrix $\bfS$ of the channel vectors. 

In \cite{haghighatshoar2016massive}, we illustrated via numerical simulations that AML Algorithm has an excellent performance for estimating user signal subspaces, especially in a JSDM system setup. Unfortunately, the SDP \eqref{eq:aml} proposed for AML Algorithm in \cite{haghighatshoar2016channel, haghighatshoar2016massive} is quite time-consuming, especially for a large array size $M$. In this paper, instead of directly solving the SDP \eqref{eq:aml}, we approximate it by another convex optimization problem for which we provide an efficient and low-complexity algorithm. Our algorithm works for the more general setup in which the $m \times M$ sampling matrices $\bfB(t)$ may be time-variant, and can be applied to more practical array configurations such as 2D rectangular arrays. 

\section{Mathematical Formulation}
\subsection{Equivalent Convex Optimization}\label{sec:eq_conv}
Let $\clG$ be a discrete grid of size $G$ consisting of the angles $\theta_i:=\sin^{-1}\big((-1+\frac{2(i-1)}{G}) \sin(\theta_{\max})\big)$, for $i\in [G]$,  over the angular range $[-\theta_{\max}, \theta_{\max}]$.  Let $\bfG=[\bfa(\theta_1), \dots, \bfa(\theta_G)]$ be an $M \times G$ matrix consisting of array responses at AoAs $\theta_i \in \clG$, $i \in [G]$. We assume that $\clG$ is dense enough such that every signal covariance matrix $\bfS$ can be well approximated by
\begin{align}\label{eq:denseness}
\bfS\approx \bfG\, \diag(s_1, \dots, s_G) \bfG^\herm=\sum_{i=1}^G s_i \bfa(\theta_i)\bfa(\theta_i)^\herm,
\end{align}
with appropriate $s_i \geq 0$, $i \in [G]$\footnote{All the results in this paper remain valid for other grids--other than $\clG$--as far as they are sufficiently dense. For $\clG$, the matrix $\bfG$ becomes an oversampled Fourier matrix, which provides the additional advantage of reducing the computational complexity as we explain in Section \ref{sec:comp_complexity}.}.  For a ULA of size $M$, taking $G\approx 2M$ is typically sufficient for this approximation to hold. 
Also, note that any $M\times M$ matrix of the form \eqref{eq:denseness} is a valid channel covariance matrix since it corresponds to the covariance matrix of the channel vector $\bfh(t)=\sum_{i=1}^G w_i(t) \bfa(\theta_i)$ consisting of $G$ scatterers with channel gains $w_i(t)\sim\cg(0, s_i)$ located at AoA $\theta_i$.

Now, consider the following convex optimization problem to be solved for the $G\times T$ matrix $\bfW$:
\begin{align}\label{eq:l2_1}
\bfW^*=\argmin_{\bfW}\frac{1}{2}\|\check{\bfG} \bfW - \bfX\|^2 + \sqrt{T}  \|\bfW\|_{2,1},
\end{align}
where $\|\bfW\|_{2,1}$ denotes the $l_{2,1}$-norm of $\bfW$ defined by $\|\bfW\|_{2,1}:=\sum_{i=1}^G \|\bfW_{i,.}\|$, where $\bfX=[\bfx(1), \dots, \bfx(T)]$ is the $m\times T$ matrix of  noisy sketches, and where $\check{\bfG}=\frac{1}{\sqrt{m}} \bfB \bfG$ is an $m\times G$ matrix with columns of unit $l_2$-norm. We prove the following result.
\begin{proposition}\label{convex_eq}
Assume that the grid $\clG$ is dense enough such that every covariance matrix $\bfS$ can be well approximated according to \eqref{eq:denseness}. Then, the SDP \eqref{eq:aml} and the convex optimization \eqref{eq:l2_1} are equivalent, in the sense that if  $\bfW^*$ is the minimizer of \eqref{eq:l2_1}, then the optimal solution of \eqref{eq:aml} can be approximated by $\bfS^*={\bfG} \,\diag(s^*_1, \dots, s^*_G) {\bfG}^\herm$, where $s^*_i=\frac{\|\bfW^*_{i,.}\|}{m\sqrt{T}}$. \hfill $\square$
\end{proposition}
\begin{proof}
Proof in Appendix \ref{app:convex_eq}.
\end{proof}
Some remarks are in order here.
\begin{remark}
Optimization problems of the type \eqref{eq:l2_1} with an $l_{2,1}$-norm regularization of the form $\tau \|\bfW\|_{2,1}$, for some regularization factor $\tau>0$, are quite well-known and are widely applied to solve \textit{Multiple Measurement Vector} (MMV) problem in Compressed Sensing \cite{tropp2006algorithms2, malioutov2005sparse}, where multiple measurement vectors correspond to different realizations of a sparse vector all having the same sparsity pattern (location of nonzero coefficients). It is also well-known that $l_{2,1}$-norm regularization promotes the block sparsity or row sparsity of the optimal solution $\bfW^*$ of the channel coefficients in \eqref{eq:l2_1}. As each row $\bfW^*_{i,.}$ of $\bfW^*$ corresponds to the channel gain of a scatterer located at $\theta_i \in \clG$ across $T$ training slots, considering the sparse scattering channel in the angular domain, this seems to be quite a reasonable regularization. However, the main novelty in \eqref{eq:l2_1} consists of the remarkable fact that for the particular choice $\lambda = \sqrt{T}$ of the regularization coefficient and within the assumptions of Proposition \ref{convex_eq},  this particular instance of MMV is, asymptotically for sufficiently dense angular grids, equivalent to the AML Algorithm, which is derived in a completely different way without any assumption on grid quantization in the angular domain. \hfill $\lozenge$
\end{remark}

\begin{remark}\label{rem:G_large}
Proposition \ref{convex_eq} implies that by increasing the number of grid points $G$ the optimal solution $\bfS^*$ of SDP \eqref{eq:aml} can be better approximated as ${\bfG} \,\diag(s_1, \dots, s_G) {\bfG}^\herm$, with appropriate $s_i >0$. However, by increasing $G$ the columns of the matrix $\bfG$, containing array responses over the grid points, become more and more correlated. It is well known from classical Compressed Sensing \cite{donoho2006compressed, candes2006near} that in a sparse estimation problem such as \eqref{eq:l2_1}, the correlation among the columns of the sensing matrix typically degrades the performance of estimation of $\bfW^*$, e.g., it creates spurious rows in $\bfW^*$. It is remarkable that, as far as estimating the signal subspace $\bfS^*$ is concerned, increasing $G$ does not incur any degradation of the performance, thanks to the convergence of \eqref{eq:l2_1} to \eqref{eq:aml} proved in Proposition \ref{convex_eq}. \hfill $\lozenge$
\end{remark}

\subsection{Forward-Backward Splitting}
In this section, we derive our low-complexity algorithm for solving the optimization problem \eqref{eq:l2_1} using the well-known \textit{Forward-Backward Splitting} (FBS) for minimizing sum of two convex functions (see \cite{combettes2005signal, combettes2011proximal} and the refs. therein).
\begin{definition}\label{def:prox}
Let $g:\bC^k \to \bR$ be a convex function. The proximal operator of $g$ denoted by $\prox_g: \bC^k \to \bC^k$ is defined by $\prox_g(\bfx):=\argmin _{\bfy \in \bC^k} g(\bfy) + \frac{1}{2} \|\bfx -\bfy\|^2$. \hfill $\lozenge$
\end{definition}
Note that for any arbitrary convex function $g$ and a fixed $\bfx \in \bC^k$, the modified convex function $g(\bfy) + \frac{1}{2} \|\bfx -\bfy\|^2$ is strongly convex and has a unique minimizer, thus, $\prox_g(\bfx)$ is always well-defined (single-valued) for any arbitrary $\bfx \in \bC^k$.

Consider the objective function in \eqref{eq:l2_1}. After suitable scaling, we can write \eqref{eq:l2_1} as the minimization of the convex function $f(\bfW)=f_1(\bfW)+f_2(\bfW)$, where 
\begin{align}
f_1(\bfW):=\frac{1}{2\zeta}\|\check{\bfG} \bfW - \bfX\|^2,\ f_2(\bfW):= \|\bfW\|_{2,1},
\end{align}
where $\zeta=\sqrt{T}$. The gradient of $f_1$ is given by $\nabla f_1(\bfW)=\frac{1}{\zeta}\check{\bfG}^\herm (\check{\bfG} \bfW - \bfX)$. Notice that $\nabla f_1$ is a Lipschitz function with a Lipschitz constant $\beta$, i.e., 
\begin{align}
\|\nabla f_1(\bfW) - \nabla f_1(\bfW')\| \leq \beta \|\bfW - \bfW'\|,
\end{align}
with $\beta=\frac{1}{\zeta}\lambda_{\max} (\check{\bfG}^\herm \check{\bfG})=\frac{1}{\zeta}\lambda_{\max} (\check{\bfG} \check{\bfG}^\herm)$, where $\lambda_{\max}$ denotes the maximum singular value of a given matrix. Note that if the grid size $G$ is sufficiently large and the grid points are distributed approximately uniformly over the AoAs, we have that 
\begin{align*}
\check{\bfG}\check{\bfG}^\herm&=\frac{1}{m} \bfB \Big \{ \sum_{i=1}^G \bfa(\theta_i) \bfa(\theta_i)^\herm\Big \}\bfB^\herm\approx \frac{G}{m} \bfB \bfI_M \bfB^\herm= \frac{G}{m} \bfI_m,
\end{align*}
which implies that $\beta=\frac{G}{\zeta m}=\frac{G}{ m \sqrt{T}}$. Using the standard results, we obtain the following upper bound for $f_1(\bfW)$.
\begin{proposition}\label{sorrogate}
Let $\bfW'$ be a given point. Then, $f_1(\bfW)$ for every $\bfW$ can be upper bounded by $\check{f}_1(\bfW)$, where
\begin{align*}
\check{f}_1(\bfW)=f_1(\bfW') +\inpr{\nabla f_1(\bfW')}{\bfW-\bfW'} + \frac{\beta}{2} \|\bfW - \bfW'\|^2,
\end{align*}
where $\inpr{.}{.}$ denotes the real-valued inner product.  \hfill $\square$
\end{proposition}
\begin{proof}
Proof in Appendix \ref{app:sorrogate}.
\end{proof}
From Proposition \ref{sorrogate}, it follows that $\check{f}_1(\bfW)$ gives an upper bound on $f_1(\bfW)$ around a given point $\bfW'$, which is indeed tight at $\bfW'$. This implies that $f(\bfW)$ can be upper-bounded by $\check{f}(\bfW):=\check{f}_1(\bfW)+f_2(\bfW)$. 
Minimizing $\check{f}(\bfW)$ can be equivalently written as minimizing $f_2(\bfW)+ \frac{\beta}{2} \| \bfW-\bfW' +\frac{1}{\beta} \nabla f_1(\bfW')\|^2$, whose optimal solution is given by $\bfW''=\prox_{\frac{1}{\beta} f_2}(\bfW'- \frac{1}{\beta} \nabla f_1(\bfW'))$ in terms of the proximity operator of the $l_{2,1}$-norm $f_2(\bfW)=\|\bfW\|_{2,1}$ according to Definition \ref{def:prox}. Standard calculations show that for a given $\alpha>0$, $(\prox_{\alpha f_2}(\bfW))_{i,.}= \frac{(\|\bfW_{i,.}\|-\alpha)_+}{\|\bfW_{i,.}\|} \bfW_{i,.}$ is obtained by simply shrinking the rows of $\bfW$, where  $(x)_+:=\max(x,0)$.

With this explanation, we propose the following iterative algorithm based on FBS. We initialize $\bfW^{(0)}={\bf 0}$ and define for $k=1, 2, \dots$ the sequence $\bfW^{(k+1)}:= \prox_{\frac{1}{\beta} f_2} \big ( \bfW^{(k)} - \frac{1}{\beta} \nabla f_1(\bfW^{(k)}) \big )$. It is seen that the estimate $\bfW^{(k+1)}$ at iteration $k$ is obtained by lower-bounding the function $f_1(\bfW)$  in a neighborhood of $\bfW'=\bfW^{(k)}$ by $\check{f}_1(\bfW)$ according to Proposition \ref{sorrogate} and finding the optimal solution of the resulting function $\check{f}(\bfW)$.  A variable step-size variant of our proposed algorithm is given in Algorithm \ref{fb_alg}, where  
in each iteration, the functions $f_1$ is minimized by moving along $-\nabla f_1$ with a positive step-size $\alpha_k$ (forward step), followed by $f_2$ being minimized by applying the proximal operator $\prox_{\alpha_k f_2}$ (backward step). This approach is known as \text{operator splitting} since individual components of $f$, i.e., $f_1$ and $f_2$, are optimized sequentially rather than jointly. The advantage is that splitting reduces the computational complexity since most of the time computing the joint proximal operator $\prox_{f_1+f_2}$ is much more complicated than computing the individual one, e.g., $\prox_{f_2}$. 

\begin{algorithm}[t]
\caption{FBS for $l_{2,1}$-Minimization.}
\label{fb_alg}
\begin{algorithmic}[1]
\State {{\bf Initialization:}} Fix $\epsilon \in (0, \min\{1, \frac{1}{\beta}\})$, $\bfW^{(0)}$.
\For{$k=1,\dots,$}
\State $\alpha_k\in [\epsilon, 2/\beta-\epsilon]$
\State $\bfZ^{(k)}=\bfW^{(k)} - \alpha_k \nabla f_1(\bfW^{(k)})$
\State $\chi_k \in [\epsilon,1]$\;
\State $\bfW^{(k+1)}=\bfW^{(k)} + \chi_k (\prox_{\alpha_k f_2}(\bfZ^{(k)}) -\bfW^{(k)})$.
\EndFor
\end{algorithmic}
\end{algorithm}

As $f(\bfW)$ is strongly convex, it has a unique optimal solution $\bfW^*$. From the convergence analysis in \cite{combettes2005signal}, we obtain the following result.
\begin{proposition}
Let $\{\bfW^{(k)}\}_{k=0}^\infty$ be the sequence generated by Algorithm \ref{fb_alg} for an arbitrary initial point $\bfW^{(0)}$ and for arbitrary selection of step-sizes according to Algorithm \ref{fb_alg}. Then, $\{\bfW^{(k)}\}_{k=0}^\infty$ converges to the unique solution $\bfW^*$. \hfill $\square$
\end{proposition}

In order to further increase the convergence speed of Algorithm \ref{fb_alg}, we apply Nestrov's update rule \cite{nesterov1983method}, which has been applied for the $l_1$-norm minimization in \cite{beck2009fast}.
\begin{algorithm}[h]
\caption{FBS with Nestrov's Update.}
\label{fb_alg_nest}
\begin{algorithmic}[1]
\State {{\bf Initialization:}} Fix $\bfW^{(0)}$, set $\bfZ^{(0)}=\bfW^{(0)}$, and $t_0=1$.
\For{$k=0,1,\dots,$}
\State $\bfR^{(k)}=\bfZ^{(k)} - \frac{1}{\beta} \nabla f_1(\bfZ^{(k)})$.
\State $\bfW^{(k+1)}=\prox_{\frac{1}{\beta} f_2} (\bfR^{(k)})$.
\State $t_{k+1}=\frac{1+\sqrt{4 t_k^2+1}}{2}$.
\State $\alpha_k=1+\frac{t_k -1}{t_{k+1}}$.
\State $\bfZ^{(k+1)}=\bfW^{(k)} + \alpha_k(\bfW^{(k+1)} - \bfW^{(k)})$.
\EndFor
\end{algorithmic}
\end{algorithm}
\begin{proposition}[Theorem 11.3.1 in  \cite{nemirovski2005efficient}]\label{prop:nest}
Let $\{\bfW^{(k)}\}_{k=0}^\infty$ be the sequence generated by Algorithm \ref{fb_alg_nest} for an arbitrary initial point $\bfW^{(0)}$ and for the step-sizes according to the Nestrov's update rule. Then, for any $k$, we have $f(\bfW^{(k+1)}) - f(\bfW^{*}) \leq \frac{4 \beta \|\bfW^*-\bfW^{(0)}\|^2}{(k+1)^2}$. \hfill $\square$
\end{proposition}

\begin{remark}
The advantage of Nestrov's update, as seen from Proposition \ref{prop:nest}, is that the gap to the optimal value, i.e., $f(\bfW^{(k)}) - f(\bfW^{*})$, scales like $O(\frac{1}{k^2})$ as a function of the number of iterations $k$ rather than $O(\frac{1}{k})$ that typically occurs for the selection of step-sizes according to Algorithm \ref{fb_alg}. In particular, the scaling $O(\frac{1}{k^2})$ is optimal \cite{nemirovski2005efficient}. \hfill $\lozenge$
\end{remark}

\begin{remark}\label{rem:slow}
As mentioned in Remark \ref{rem:G_large}, increasing the grid size $G$ does not degrade the performance of the subspace estimation. However, since the Lipschitz constant $\beta=\frac{G}{m \sqrt{T}}$ grows proportionally to $G$, it is seen from Proposition \ref{prop:nest} that increasing $G$ reduces the  speed of  the algorithm. \hfill $\lozenge$
\end{remark}

\subsection{Time-Varying Sampling Operators}\label{sec:time_varying}
When the dimension of sketches is much less than the number of antennas, i.e., $m\ll M$, or when the sampling ratio is much smaller than the normalized angular spread of channel vectors, i.e., $\frac{m}{M} \ll \frac{\Delta \theta}{2 \theta_{\max}}$, using a fixed sampling matrix $\bfB$ results in an aliasing pattern that is not typically resolved even by taking several measurements. Therefore, to improve the estimation performance, it is beneficial to use time-varying sampling matrices $\bfB(t)$ in each slot $t \in [T]$ such that the aliasing caused by a sampling matrix $\bfB(t)$ at  a time $t \in [T]$ is resolved by other sampling matrices $\bfB(t')$ at other times $t'\neq t$.
 Let $\clI_t \subseteq [M]$ denote the indices of the sampled antenna elements at time $t\in[T]$. We always assume that the indices belonging to $\clI_t$ are sorted in an increasing ordered. We follow the \matlab convention that for a vector $\bfm \in \bC^M$, we have $\bfB(t)\bfm= \bfm(\clI_t)$, where $\bfm(\clI_{t})$ denotes an $m$-dim vector containing the components of $\bfm$ belonging to $\clI_t$. All the formulations for the fixed operator $\bfB$ can be immediately extended to the time-variant case by defining $f_2(\bfW)=\|\bfW\|_{2,1}$ and $f_1(\bfW)=\frac{1}{2\zeta}  \sum_{t \in [T]} \|\check{\bfG}_t \bfW_{.,t} - \bfX_{.,t}\|^2$, where $\zeta=\sqrt{T}$ and  where $\check{\bfG}_t=\frac{1}{\sqrt{m}} \bfB(t) \bfG$, and setting $\nabla f_1(\bfW)$ to be the $G\times T$ matrix with $\nabla f_1(\bfW)_{.,t}=\frac{1}{\zeta} \check{\bfG}_t^\herm (\check{\bfG}_t \bfW_{.,t} - \bfX_{.,t})$ for $t\in[T]$. Notice that  $\nabla f_1$ is again a Lipschitz function in this case with a  Lipschitz constant 
\begin{align}
\beta&=\frac{1}{\zeta} \max_{t \in [T]} \big \{ \lambda_{\max}(\check{\bfG}_t^\herm \check{\bfG}_t) \big \}= \frac{1}{\zeta} \max_{t \in [T]} \big \{ \lambda_{\max}(\check{\bfG}_t \check{\bfG}_t^\herm) \big \}\nonumber\\
&= \frac{1}{m \zeta} \max_{t \in [T]} \big \{ \lambda_{\max}(\bfB(t) \sum_{i=1}^G \bfa(\theta_i) \bfa(\theta_i)^\herm \bfB(t)^\herm ) \big \}\nonumber\\
&\approx \frac{G}{m \zeta} \max_{t \in [T]} \big \{ \lambda_{\max}(\bfB(t) \bfI_M \bfB(t)^\herm ) \big \}\nonumber\approx\frac{G}{\zeta m}=\frac{G}{m \sqrt{T}},
\end{align}
which is the same as in the time-invariant case. This implies that all the steps of Algorithm \ref{fb_alg} and \ref{fb_alg_nest}, and their convergence guarantee still hold in this case.

\subsection{Computational Complexity}\label{sec:comp_complexity}
Each iteration of both Algorithm \ref{fb_alg} and \ref{fb_alg_nest} requires computing $T$ columns of $\nabla f_1$, where the $t$-th column, $t \in [T]$, is given by $\nabla f_1(\bfW)_{.,t}=\frac{1}{\zeta}\check{\bfG}_t^\herm (\check{\bfG}_t \bfW_{.,t} - \bfX_{.,t})$, evaluated at $\bfW=\bfW^{(k)}$ at iteration $k$. 
For the special grid $\clG$ with the discrete AoAs $\theta_i:=\sin^{-1}\big((-1+\frac{2(i-1)}{G}) \sin(\theta_{\max})\big)$, $i\in [G]$,  in the angular range $[-\theta_{\max}, \theta_{\max}]$, the matrix $\bfG$ becomes an oversampled Fourier matrix, namely, the columns of $\bfG$ are given by $(\omega_G^c, \omega_G^{2c}, \dots, \omega_G^{Mc})^\transp$, where $\omega_G=e^{j\frac{\pi}{G}}$ and where $c\in \{-G, -G+2, \dots, G-1\}$. This special structure of $\bfG$, as a result that of $\check{\bfG}$, can be exploited to compute $\nabla f_1(\bfW)$ quite efficiently.

For each $t\in [T]$, we first compute $\check{\bfG}_t \bfW_{.,t}$. Following the \matlab notation, let ${\bfm}=G\, \mathsf{ifft}(\bfW_{.,t},G) \in \bC^G$ be the inverse \textit{Discrete Fourier Transform} (DFT) of $\bfW_{.,t}$ scaled with $G$, which can be efficiently computed using the \textit{Fast Fourier Transform} (FFT) algorithm provided that $M$ and $G$ are powers of $2$. Then, $\check{\bfG}_{t} \bfW_{.,t}$ is simply given by $\frac{1}{\sqrt{m}} {\bfm}(\clI_{t})$, where $\clI_{t}$ denote the indices of the sampled antennas at $t \in [T]$. The whole complexity of this step of calculation for all $t\in [T]$ is $O\big(T G \log_2(G) \big )$. 
After computing $\check{\bfG}_t \bfW_{.,t}$, $t \in [T]$, we need to calculate $\check{\bfG}_t^\herm \bfr_{t}$, where $\bfr_{t}=\check{\bfG}_{t} \bfW_{.,t} - \bfX_{.,t}$. This can be simply done by setting $\bfm$ to be an $M$-dim all-zero vector, and embedding $\bfr_t$ in $\bfm$ in indices belonging to $\clI_t$ such that $\bfm(\clI_{t})=\bfr_{t}$ and taking the $G$-point DFT of $\bfm$, which gives $\check{\bfG}_t^\herm \bfr_{t}=\frac{1}{\sqrt{m}} \mathsf{fft}(\bfm, G)$. The whole complexity of this step for all $t \in [T]$ is again $O\big(T G \log_2(G) \big )$.

Letting $T_\mathsf{conv}$ be the number of iterations necessary for the convergence, the whole computational complexity is $O\big(2 T_\mathsf{conv} T G \log_2(G) \big )$, which is at least two orders of magnitude less than the complexity of directly solving the SDP \eqref{eq:aml} with off-the-shelf SDP solvers. As we explained in Remark \ref{rem:G_large}, increasing the grid size $G$ does not degrade the recovery performance. However, as also mentioned in Remark \ref{rem:slow}, it increases the Lipschitz constant $\beta$ of $\nabla f_1$ and slows down the convergence of the algorithm. The main reason is that increasing $\beta$ makes the shrinkage operation in the proximal operator $\prox_{\frac{1}{\beta} f_2}$ softer. As a result, the algorithm requires more iterations to identify the dominant grid elements. Thus, we expect that $T_{\mathsf{conv}}$ scale proportionally to the oversampling factor $\frac{G}{M}$. We always use $\frac{G}{M}=2$. Our numerical simulations show that for this choice of oversampling factor, both Algorithm \ref{fb_alg} and \ref{fb_alg_nest}, and especially Algorithm \ref{fb_alg_nest},  converge in only a couple of iterations.

\section{Extension to Other Array Geometries}
\subsection{2D Rectangular Array Configurations}
Our proposed algorithms can be extended to a 2D rectangular array consisting of $M=M_xM_y$ antenna elements, arranged over a rectangular grid%
 \small
\begin{align*}
\clR=\Big\{\big ((i-\frac{M_x+1}{2})d_x, (j - \frac{M_y+1}{2}) d_y\big): i \in [M_x], j \in [M_y]\Big \},
\end{align*}\normalsize
in the 2D plane of the array, having a horizontal spacing $d_x$ and a vertical spacing $d_y$ between its elements. We consider a 3D Cartesian coordinate chart with an $xy$-plane given by the 2D plane of the array and with a $z$-axis orthogonal to it. 
We denote the $M$-dim ($M=M_xM_y$) array responses by $\bfa(\xi)$, where $\xi$ belongs to the unit 2D sphere $\bS^2=\{\xi \in \bR^3: \|\xi\|=1\}$ (lying in 3D space) and parameterizes the AoAs; sometimes, it is better to use a coordinate chart for $\bS^2$ in which every point $\xi$ is represented by two angles: the polar angle $\theta$ and the azimuthal angle $\phi$. It is more convenient to denote the $M$-dim  array response $\bfa(\xi)$ with double index $(x,y) \in [M_x]\times [M_y]:=\{(i_x,i_y): i_x \in [M_x], i_y\in [M_y]\}$, where we have
\begin{align}
[\bfa(\xi)]_{xy}=e^{j \frac{2\pi}{\lambda} \inp{\xi}{\bfr_{xy}}},
\end{align}
where $\bfr_{xy}=\big ( (x-\frac{M_x+1}{2}) d_x, (y-\frac{M_y+1}{2}) d_y,0\big )\in \bR^3$ denotes the location of the array element indexed by $(x,y)$ in the 2D plane of the array ($xy$-plane).
The channel vector of a user, whose scattering channel consists of a collection of $p$ scatterers with AoAs parameterized by $\{\xi_i: i \in [p]\}$ and channel gains $\{w_i(t): i \in [p]\}$, is given by $\bfh(t)=\sum_{i=1}^p w_i(t)\bfa(\xi_i)$.  The channel covariance matrix is also given by $\bfS=\sum_{i=1}^p \sigma_i^2  \bfa(\xi_i) \bfa(\xi_i)^\herm$, which using the double-index notation can be represented by 
\begin{align}\label{eq:blt_not}
[\bfS]_{xy,x'y'}&=\sum_{i=1}^p \sigma_i^2 [\bfa(\xi_i)]_{xy} [\bfa(\xi_i)^\herm]_{x'y'}\nonumber\\
&=\sum_{i=1}^p \sigma_i^2 e^{j \frac{2\pi}{\lambda} \inp{\xi}{\bfr_{xy}-\bfr_{x'y'}}}\nonumber\\
&=\sum_{i=1}^p \sigma_i^2 e^{j \frac{2\pi}{\lambda} \inp{\xi}{\bfr_{x-x',y-y'}}}.
\end{align}
It is seen from \eqref{eq:blt_not} that $\bfS$ has a block-Toeplitz form, i.e., it can be represented by an $M\times M$ matrix containing $M_x\times M_x$ blocks of dimension $M_y\times M_y$, where the matrix at block $(x,x')$ is given by $\scrU_{x-x'}$ and depends only on $x-x'$, where we also have $\scrU_{-k}=\scrU_{k}^\herm$, $k \in [M_x]$, due to the Hermitian symmetry. Similarly, $\bfS$ can be represented with an $M\times M$ matrix containing $M_y\times M_y$ blocks of dimension $M_x\times M_x$, where the matrix at block $(y,y')$ is given by $\scrV_{y-y'}$, with $\scrV_{-k}=\scrV_{k}^\herm$, $k \in [M_y]$. In fact, $\bfS$ is even more structured since all the diagonal blocks of $\bfS$ in both block representations are equal to a Toeplitz matrix, whereas a block-Toeplitz matrix generally might not have Toeplitz diagonal blocks. The originally proposed AML Algorithm for the ULA in \cite{haghighatshoar2016channel, haghighatshoar2016massive} can be generalized to 2D rectangular arrays. It can be formulated as an SDP similar to \eqref{eq:aml} by replacing the constraint set $\clT_+$ with the set of PSD Hermitian block-Toeplitz matrices denoted by $\clB\clT_+$, which is still a convex set.  

We again assume that in each slot $t\in [T]$, we only sample a collection of $m\ll M$ array elements via a possibly time-variant sampling matrix $\bfB(t)$. Similar to the previous case for ULA, we define a 2D grid $\clG$ of size $G$ by quantizing the continuum of AoAs, and construct the $M\times G$ matrix consisting of the array responses over the discrete AoAs belonging to $\clG$. A direct inspection in the proof of Proposition \ref{convex_eq} indicates that the SDP for AML Algorithm in this case can still be approximated by the $l_{2,1}$-norm regularized convex optimization in \eqref{eq:l2_1}.
All the steps of the algorithm and all the parameters remain the same as in the case of ULA.  However, due to the 2D lattice array configuration, we need to apply 2D DFT to compute $\nabla f_1(\bfW)$ in each step rather than 1D DFT used for the ULA. This can still be efficiently computed provided that both $M_x$ and $M_y$, and the oversampling ratios $\frac{G_x}{M_x}$ and $\frac{G_y}{M_y}$ are powers of $2$, where the total computational complexity is again given by  $O\big(2 T_\mathsf{conv} T G \log_2(G) \big )$. 

As explained in Section \ref{sec:comp_complexity}, using the computational advantage of FFT algorithm requires a special design of the grid $\clG$ that contains the AoAs $\theta_i=\sin^{-1}\big((-1+\frac{2(i-1)}{G}) \sin(\theta_{\max})\big)$, $i\in [G]$. However, if $G$ is large enough, $\clG$ has a performance comparable with any other grid of similar size in approximating the signal covariance matrix (see \eqref{eq:denseness}) since it covers the whole angular range $[-\theta_{\max}, \theta_{\max}]$. Unfortunately, this is not the case for 2D rectangular arrays: exploiting the computational advantage of 2D FFT restricts the range of AoAs that can be processed. To explain this better, let us consider two uniform grids: $\clG_x$ a grid of size $G_x$ in $[-\xi^x_{\max}, \xi^x_{\max}]$ and $\clG_y$ a grid of size $G_y$ in $[-\xi^y_{\max}, \xi^y_{\max}]$, where $\xi^x_{\max}$ and $\xi^y_{\max}$ are such that $\xi^x_{\max} \overline{d}_x=\xi^y_{\max} \overline{d}_y=1$, where we define $\overline{d}_x=d_x/(\lambda/2)$ and $\overline{d}_y=d_y/(\lambda/2)$ as the normalized horizontal and vertical spacing between the array elements in the 2D grid $\clR$. 
We also assume that $\xi^x_{\max}$ and $\xi^y_{\max}$ satisfy the additional constraint $(\xi^x_{\max})^2+ (\xi^y_{\max})^2 \leq 1$. Let us consider the following grid consisting of $G=G_xG_y$ points on the unit sphere $\bS^2$, each representing a specific AoA:
\begin{align}\label{eq:uni_grid_def}
\clG=\{(\xi_x,\xi_y, \sqrt{1-\xi_x^2-\xi_y^2}): \xi_x \in \clG_x, \xi_y \in \clG_y\}.
\end{align}
This is illustrated in Fig.\,\ref{fig:2D_grid} with the grid points lying on the unit sphere. It is seen that $\clG$ can  cover only a subset of all possible AoAs.
\begin{figure}[t]
\centering
\includegraphics{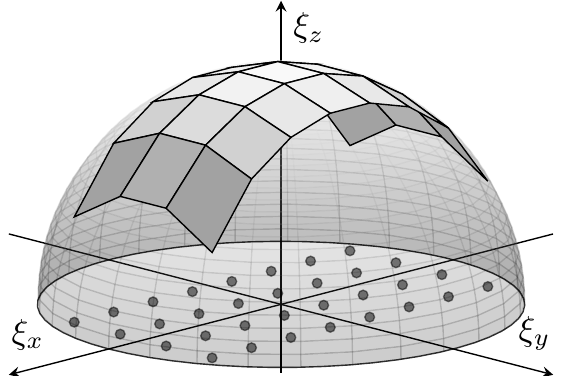}
\caption{A non-uniform grid of AoAs over the unit sphere, whose projection on the $\xi_x\xi_y$-plane is a rectangular grid.}
\label{fig:2D_grid}
\end{figure}
From \eqref{eq:uni_grid_def} and the definition of $\clG_x$ and $\clG_y$, it is not difficult to check that the projection of the grid points on the $\xi_x\xi_y$-plane builds a 2D rectangular grid enclosed by the rectangle $[-\xi^x_{\max}, \xi^x_{\max}]\times [-\xi^y_{\max}, \xi^y_{\max}]$. Also, notice that the array response $\bfa(\xi)$ at a $\xi \in \clG$ is given by $[\bfa(\xi)]_{xy}=e^{j \pi (x-\frac{M_x+1}{2})\overline{d}_x\xi_x} e^{j \pi (y-\frac{M_y+1}{2})\overline{d}_y\xi_y}$, where $x\in[M_x]$, $y \in [M_y]$, and $\xi_x \in \clG_x$ and $\xi_y\in \clG_y$ denote the $xy$ component of $\xi$.  Since $x$, $y$, $\xi_x$ and $\xi_y$ all take values in discrete lattices (with uniform spacing $1$, $1$, $1/\overline{d}_x$ and $1/\overline{d}_y$ respectively), letting $\bfG$ be the $M\times G$ matrix whose columns are given by $\bfa(\xi)$, $\xi \in \clG$, with a suitable ordering, we obtain the 2D DFT matrix. If $M_x$ and $M_y$ are powers of $2$ and the oversampling ratios $\frac{G_x}{M_x}$ and $\frac{G_y}{M_y}$ are also powers of two, similarly to the case of ULA in Section \ref{sec:comp_complexity}, we can apply the 2D FFT algorithm to compute $\nabla f_1$ quite fast.

\begin{remark}
Although $\xi^x_{\max} \overline{d}_x$ and $\xi^y_{\max} \overline{d}_y$ could be selected to be less than $1$, this unreasonably restricts the spatial resolution of the array. For a practical design, we should first decide on the subset of AoAs on the unit sphere that we intend to process, with the additional constraint that the projection of this subset on the $\xi_x\xi_y$-plane must lie in a symmetric rectangular region, which is necessary in order to take advantage of computational benefits of 2D FFT. This yields the desired $\xi^x_{\max}$ and $\xi^y_{\max}$. Finally, to obtain the best spatial resolution in the desired region, the array spacings $d_x$ and $d_y$ should be set to their maximum values such that $\xi^x_{\max} \overline{d}_x=\xi^y_{\max} \overline{d}_y=1$. \hfill $\lozenge$
\end{remark}

\subsection{General Array Configurations}
Consider a general array configuration, in which the array responses are parameterized by $\{\bfa(\xi): \xi \in \Xi\}$, for some parameter set $\Xi$ representing the AoAs. In this case, the space of all feasible signal covariance matrices is given by the set of all $M\times M$ matrices $\clS:=\{\int_{\Xi} \gamma(d\xi) \bfa(\xi) \bfa(\xi)^\herm: \gamma \in \clP(\Xi) \}$ where $\clP(\Xi)$ denotes the space of all positive measures over $\Xi$. Notice that $\clS$ is indeed a convex subset (cone) of the cone of all $M\times M$ PSD matrices. Depending on the array geometry, it might happen that $\clS$ has a simple algebraic representation that can be exploited in the optimizations. For example, $\clS$ coincides with the space of PSD Hermitian Toeplitz matrices $\clT_+$ for the ULA, and with the space of PSD Hermitian block-Toeplitz matrices $\clB\clT_+$ for a 2D lattice array configuration.

The SDP formulation \eqref{eq:aml} for the AML Algorithm can be extended to this case by replacing $\clT_+$ with $\clS$. In particular, the equivalence between SDP \eqref{eq:aml} and the $l_{2,1}$-regularized convex optimization \eqref{eq:l2_1} still holds provided that the set $\Xi$ is quantized with a sufficiently dense grid $\clG$, such that $\sup_{\xi \in \Xi} \inf_{ \xi' \in \clG} \|\bfa(\xi)\bfa(\xi)^\herm - \bfa(\xi')\bfa(\xi')^\herm\| \leq \epsilon M$ holds for a sufficiently small $\epsilon \in (0,1)$. Due to the iso-norm property of the array response vectors $\{\bfa(\xi): \xi \in \Xi\}$, this condition is satisfied provided that $\sup_{\xi \in \Xi} \inf_{ \xi' \in \clG} \|\bfa(\xi) - \bfa(\xi')\| \leq \epsilon \sqrt{M}$. Overall, we can not expect to have the low  per-iteration computational complexity $O(TG \log_2(G))$ obtained because of using the FFT algorithm in the case of the ULA or the 2D lattice array  unless the covariance matrices in $\clS$ have other special algebraic structures that can be exploited to speed up numerical computations. Furthermore, one might also need to restrict the range of AoAs, as in the 2D lattice configuration, to benefit this underlying algebraic structure.

\section{Subspace Tracking}\label{sec:tracking}
\subsection{Extending the Algorithm to the Tracking mode}
Up to now, we have assumed that, although the channel gains $\bfw(t)$ as in \eqref{eq:disc_ch_mod}, and as a result the channel vectors $\bfh(t)$ vary i.i.d. with time, the underlying channel geometry $\{(\sigma_i^2, \theta_i): i \in [p]\}$ embedded in the covariance matrix $\bfS=\sum_{i=1}^p \sigma_i^2 \bfa(\theta_i) \bfa(\theta_i)^\herm$ remains stable for quite a long time, especially much longer than the window size $T$. This allows the signal subspace to be estimated from the low-dim sketches inside the window $[T]$, and to be used for the rest of time. In practice, $\{\bfh(t)\}_{t=1}^\infty$ as a stochastic process is only locally stationary and its statistics (covariance matrix) is piecewise constant, i.e., constant over 
rather long intervals of time (time scale of one to tens of seconds) and changes with abrupt transitions when the scattering environment of the user changes (e.g., while moving form indoor to outdoor or turning from one street to another for a moving vehicle). In any case, the duration of the time intervals over which the covariance is time-invariant is $3$ up to $4$ orders of magnitude larger than the duration of the data transmission slots. Therefore, we can collect a window of $T$ i.i.d. samples (for a sufficiently large $T$) in the time-frequency domain over each interval \cite{mahler2015propagation}.

Traditionally, there are two approaches in the literature to deal with sharp transitions in signal statistics: 1)\,change point detection (see \cite{poor2009quickest} and refs. therein) and 2)\,online tracking (see \cite{comon1990tracking, yang1995projection, crammer2006online, balzano2010online, chi2013petrels, sayed2003fundamentals} and the refs. therein). Adapted to the subspace estimation in this paper, in the former, one applies change point detection algorithms to identify the transition points in the statistics, and upon identifying a transition point, the subspace estimation algorithm is run to reestimate/update the signal subspace from new observations. The resulting estimate is used until the next transition point is identified. In the latter, in contrast, upon receiving a new observation (sketch) $\bfx(t)$ at time $t$, the tracking algorithm updates its estimate of the signal subspace $\bfS(t)$ by 
\begin{align}\label{eq:subspace_update}
\bfS({t+1})=\alpha \bfS(t) + (1-\alpha) \scrI(\bfx(t)),
\end{align}
where $\alpha \in (0,1)$ is the update factor, and where $\scrI(\bfx(t))$ is a subspace innovation term that depends on the newly received sketch $\bfx(t)$ (see \cite{comon1990tracking, yang1995projection, crammer2006online, balzano2010online, chi2013petrels} and the refs. therein). The choice of $\alpha$ makes a trade-off between the quality of the subspace estimation in the stationary regime (variance) and the tracking ability of the algorithm in the non-stationary transition regime (bias)\footnote{This is well known as the bias-variance trade-off in statistics.}: the closer $\alpha$ to $1$, the less variance in subspace estimation in the stationary regime, and the closer $\alpha$ to $0$ the faster the subspace identification after occurring a sharp transition in the non-stationary transition region. Broadly speaking,  for a given $\alpha \in (0,1)$ the data used for subspace estimation effectively comes from a window of   $T_\alpha \approx \frac{1}{\log(\frac{1}{\alpha})}$  latest observations. In fact, $\alpha \to 1$ makes $T_\alpha$ larger and improves the subspace estimation provided that the window of observations lies in a stationary regime. However, increasing $T_\alpha$ also increases the probability of having a  transition in the middle of the window, in which case the subspace estimation algorithm requires around $T_\alpha$ new observations to identify the new subspace after the transition, thus, making the tracking algorithm less agile.

Our proposed algorithms can be run in the tracking mode as follows. We fix a window size $T$, which corresponds to selecting a suitable value for the tracking parameter $\alpha \in (0,1)$ in the tracking algorithm. At every time $t$, we always keep the latest $T$ sketches $\clW_t:=\{\bfx(t-T+1), \dots, \bfx(t)\}$ and update it upon receiving a new sample $\bfx(t+1)$ as $\clW_{t+1}=\clW_t \cup \{\bfx(t+1)\} \backslash \{\bfx(t-T+1)\}$. We use the optimal solution $\bfW^*(t)$ of the convex optimization \eqref{eq:l2_1}, when the matrix of sketches is set to $\bfX=\clW_t$, as a \textit{warm} initialization to the algorithm at time $t+1$. Typically $T \gg 1$, and we expect that adding the new sketch $\bfx(t+1)$ does not effect the optimal solution considerably. In fact, for a window size $T$, we expect intuitively that $\frac{\|\bfW^*(t+1) - \bfW^*(t)\|^2}{\|\bfW^*(t)\|^2}=O(\frac{1}{T})$. On one hand, this implies that, for a large $T$, only $O(1)$ number of iterations would be sufficient to reach from the  old estimate $\bfW^*(t)$ (used as the initialization point) to the new estimate $\bfW^*(t+1 )$, thus, the whole complexity of the subspace update would be of the order $O(2 G \log_2(G))$ per each newly arrived observation. On the other hand, this indicates that, as expected, increasing the window size $T$, makes the algorithm less agile to sharp subspace transitions since the new estimate $\bfW^*(t+1)$ can not move far from the old one $\bfW^*(t)$ in a single iteration.

\subsection{Further Simplified Subspace Tracking}\label{simple_tracking}
In practical implementations, updating the subspace at each time $t$ requires the following steps: updating the weighting matrix $\bfW(t)$, computing the estimate of signal covariance matrix according to Proposition \ref{convex_eq}, and computing the SVD of the resulting covariance matrix and identifying its dominant subspace. This might be too complicated in some real-time implementations in massive MIMO. Instead, we can use the weighting matrix $\bfW(t)$ to identify the position of dominant elements (active elements) in the over-complete dictionary over the grid given by $\bfG$. This only requires updating the $l_2$-norm or equivalently the $l_2$-norm squared of the rows of $\bfW(t)$, i.e., $\|\bfW(t)_{i,.}\|^2$ for $i\in [G]$, after each iteration $t$, which can be done quite fast. At each time $t$, the $M\times q$ submatrix of $\bfG$ (for some $q \ll M$) corresponding to the dominant rows of $\bfW(t)$ with significantly large $l_2$ norms  provides an estimate of the dominant signal subspace.

\section{Simulation Results}\label{sec:sim}
In this section, we evaluate the performance of our proposed subspace estimation/tracking algorithm empirically via numerical simulations. 
\subsection{Array Model}\label{sec:arr_model}
We consider a ULA with $M=64$ antennas, where in each training period, we randomly sample only $m=16$ of them, thus, a sampling ratio of $\rho:=\frac{m}{M}=0.25$. We assume that the array has $\theta_{\max}=60$ degrees and scans an angular range of $\Delta \theta_{\max}=2\theta_{\max}=120$ degrees. For all the simulations, we use a grid size of $G=2M$, thus, a grid oversampling factor of $2$, and optimize \eqref{eq:l2_1} by running Algorithm \ref{fb_alg_nest} with the  Nestrov's update rule and with i.i.d. time-varying sampling matrices explained in Section \ref{sec:time_varying}.

\subsection{Scaling with respect to Training Signal-to-Noise Ratio}\label{sim:scale_snr}
We consider a scattering geometry, in which the received signal power of a given user is uniformly distributed over the angular range $\Theta=[10, 30]$ degrees, with an angular spread of $20$ degrees. 
We define the training \textit{Signal-to-Noise Ratio} (SNR) by $\frac{\bE[\|\bfh(t)\|^2]}{\bE[\|\bfn(t)\|^2]}$, where $\bfh(t)$ denotes the user channel vector and where $\bfn(t)$ is the array noise at time $t$ as in \eqref{eq:disc_ch_mod}.

 Fig.\,\ref{fig:perf_se_snr} illustrates the scaling of the performance metric $\Gamma(\bfp, \widehat{\bfp})$ defined in Section \ref{sec:perf_metric} versus the training SNR for different training lengths (window sizes) $T \in \{50,100, 200\}$. To obtain the curve for each $T$, we average the resulting performance of the subspace estimator versus SNR over $100$ independent simulations, where in each simulation we run the algorithm  for approximately $50$ iterations before convergence.
 We consider two different family of measurement matrices for simulation: i)\,The 0-1 sampling matrices denoted by ``Bin'', where each row of the matrix contains one $1$ at a specific random column and corresponds to a random antenna selection as explained in Section \ref{sec:samp_op}, and ii)\,The random phase-shift matrices denoted by ``PS'', where the sampling matrices $\bfB(t)$ are generated according to $[\bfB(t)]_{r,c}=\frac{e^{j \theta_{r,c}(t)}}{\sqrt{M}} $, where $\theta_{r,c}(t)$ are selected i.i.d. randomly in each row $r$, column $c$, and across different slots $t\in [T]$ form the set of quantized phases $\{\frac{k 2\pi}{2^b}: k=0, \dots, 2^b-1\}$ with $b=5$ bit precision. In massive MIMO applications, the former family of sampling matrices can be implemented via switches whereas the latter is more suitable for implementation via constant-amplitude  phase-shafting networks.
It is seen from Fig.\,\ref{fig:perf_se_snr} that both family of sampling matrices have quite similar performances although the 0-1 sampling matrices yield a more efficient numerical implementation via the FFT algorithm as explained in Section \ref{sec:comp_complexity} by essentially avoiding any matrix multiplication.
The simulation results are qualitatively similar to the results in \cite{haghighatshoar2016channel, haghighatshoar2016massive} but are obtained for a different performance metric. In applications such as JSDM, the practically important SNR regime is around $0$\,dB up to $10$\,dB in which the system has a considerably high throughput (measured in terms of the achievable sum-rate or spectral efficiency). Fig.\,\ref{fig:perf_se_snr} illustrates that in this regime of SNR, our proposed subspace estimation scheme has an excellent performance. 

\begin{figure}[t]
\centering
\includegraphics{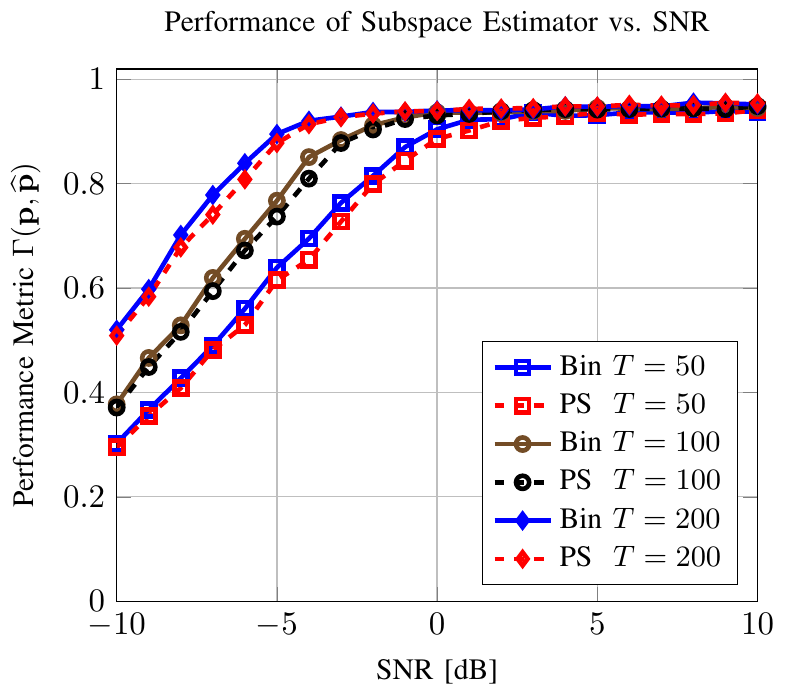}
\caption{{\small Performance of our proposed subspace estimation algorithm versus SNR for different training lengths $T\in \{50, 100, 200\}$ for 0-1 sampling matrices ``Bin'' and random phase-shift matrices ``PS''. The reason $\Gamma(\bfp,\widehat{\bfp}) \nrightarrow 1$ for large SNR is due to the finite grid size $G$.}}
\label{fig:perf_se_snr}
\end{figure}

\begin{figure}[t]
\centering
\includegraphics{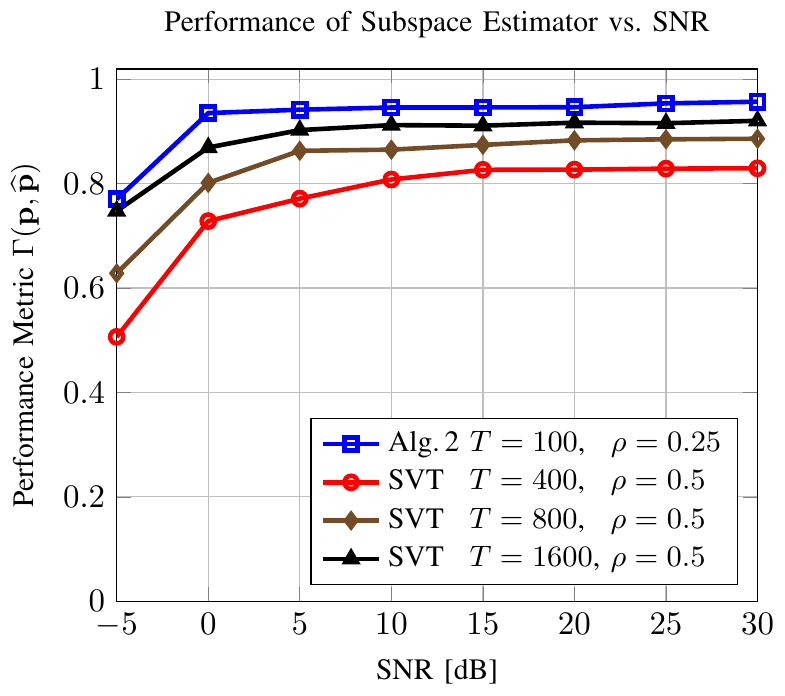}
\caption{{\small Comparison of performance versus SNR of our algorithm (Alg.\,2) for with that of SVT for different number of sketches $T$ and different sampling ratios $\rho=\frac{m}{M}$.}}
\label{fig:comp_with_svt}
\end{figure}

\subsection{Comparison with SVT Algorithm}
We compare the performance of our algorithm with that of the SVT (singular value thresholding) Algorithm in \cite{cai2010singular}, which provides the state of the art performance in subspace estimation (and matrix completion).
SVT is a \textit{batch} algorithm, i.e., it estimates the signal subspace from a fixed number of sketches $T$, in contrast with an \textit{online} tracking algorithm where the number of sketches increases with time\footnote{Although the effective number of sketches still remains constant and depends on the subspace update factor $\alpha$ as in \eqref{eq:subspace_update}.}.  Each iteration of SVT consists of computing the SVD of an $M\times T$ matrix followed by a singular value thresholding, which can be done efficiently (e.g., via Lanczos algorithm) when the matrix is quite low-rank. However, our algorithm is much faster since it does not requires any SVD computation or matrix multiplication as all the calculations are done via the FFT algorithm.

For simulations, we assume that the received signal power of the user is uniformly distributed over the angular range $\Theta=[10, 30]$ degrees as in Section \ref{sim:scale_snr}.
Fig.\,\ref{fig:comp_with_svt} illustrates the comparison of the performance versus SNR of our algorithm for $T=100$ and a sampling ratio $\rho=0.25$ with that of SVT for different $T\in \{400,800,1600\}$ and  a sampling ratio $\rho=0.5$. It is seen that  our algorithm performs much better than SVT even for a smaller data size $T$ and antenna sampling ratio $\rho$. 
In particular, for a target estimation performance, it runs orders of magnitude faster since it requires smaller $T$ and consumes much less storage (scaled proportionally to $M\times T$).

\subsection{Tracking Performance}\label{tracking_perf}
Fig.\,\ref{fig:perf_tracking_snr} illustrates the tracking performance of our proposed estimator for different window size $T\in\{50,100,200\}$ and different SNR when there is a sharp transition in the channel statistics (geometry).  We consider a time interval of length $400$ for simulation that contains a sharp transition in the middle  at  time $t_\mathsf{tr}=200$. For $t =1, \dots, t_\mathsf{tr}-1$, the user signal power is uniformly distributed in  the angular range $\Theta=[10, 30]$ degrees, with an angular spread of $20$ degrees as in Section \ref{sim:scale_snr}, whereas at time $t=t_\mathsf{tr}$, the geometry of the channel changes abruptly such that for $t=t_\mathsf{tr},\dots, 400$ the user signal power is uniformly distributed in the angular range $\Theta'=[-40,-20]$, where $\Theta \cap \Theta'=\emptyset$.

Fig.\,\ref{fig:perf_tracking_snr} illustrates a random sample path of the performance metric $\Gamma(\bfp, \widehat{\bfp})$ generated by the algorithm for the 0-1 sampling matrices during the whole simulation.  
To generate these plots, we run our proposed subspace estimation algorithm in a tracking mode as explained in Section \ref{sec:tracking}, in which upon receiving  a new sketch, we run only \textit{one} iteration of our proposed algorithm while treating the previous estimate as  initialization. We start the algorithm with the zero initialization at time $t=0$, where it is seen from Fig.\,\ref{fig:perf_tracking_snr} that the algorithm identifies the signal subspace in quite a short time. It is also seen that immediately after the sharp transition in the channel, the performance metric $\Gamma(\bfp, \widehat{\bfp}) \to 0$, however, the algorithm is able to track/identify the new signal subspace in quite a short time. Interestingly, it is seen that for a window of size $T$, the delay before identifying/tracking the new subspace is around $\frac{T}{2}$, namely, immediately after half the observation window is filled with new sketches generated with the new channel geometry, the algorithm makes a sharp transition from the old signal subspace to the new one.

\begin{figure}[t]
\centering
\includegraphics{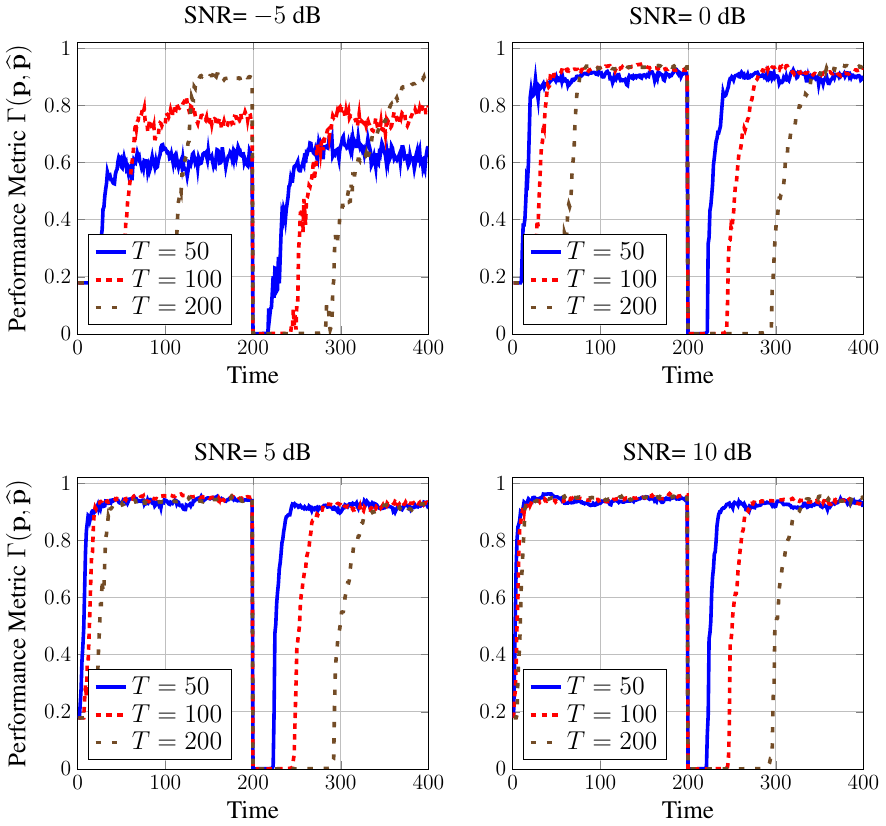}
\caption{{\small Tracking performance of our proposed algorithm (Alg.\,2) for different window sizes $T\in\{50,100,200\}$ and for different training SNR. There is a transition in statistics at time $t_{\mathsf{tr}}=200$.}}
\label{fig:perf_tracking_snr}
\end{figure}

\subsection{Comparison with PETRELS Algorithm}
We compare the performance of our algorithm in the tracking mode with that of PETRELS Algorithm proposed in \cite{chi2013petrels}. Note that PETRELS is an \textit{online} algorithm (in contrast with SVT which is a batch algorithm) and provides the state of the art performance in online subspace estimation. 
In the online mode the number of sketches is not fixed (in contrast with the batch mode) and increases with time. Denoting by $T$ the size of the window of sketches used by our proposed algorithm in the tracking mode, we set the update factor of PETRELS to $e^{-\frac{\kappa}{T}}$, i.e., $\alpha=e^{-\frac{\kappa}{T}}$ with our notation in \eqref{eq:subspace_update}, where we tune $\kappa \in [1,5]$ to obtain the best performance. In this way, we  make sure that the effective number of sketches used by PETRELS comes from a window of approximate size $T$ of the latest sketches. 

Fig.\,\ref{fig:comp_with_petrels} illustrates the simulation results. We compare the performance of PETRELS with that of our proposed tracking algorithm for different $T$ and antenna sampling ratio $\rho$. It is evidently seen that our algorithm has a much superior performance in the stationary regime and tracks the subspace transitions much faster in the non-stationary regime. Fig.\,\ref{fig:comp_with_petrels} also illustrates the interesting bias-variance trade-off underlying the PETRELS performance (which also exists in our proposed tracking algorithm as can be seen from Fig.\,\ref{fig:perf_tracking_snr}), where for a fixed sampling ratio $\rho$, the performance of PETRELS in the stationary regime (variance) improves by increasing the window size $T$, but this comes at the price of much slower tracking in the non-stationary regime. It is also seen from Fig.\,\ref{fig:comp_with_petrels} that the performance of PETRELS degrades considerably by reducing the antenna sampling factor $\rho$, whereas our algorithm is less sensitive to $\rho$ and works perfectly even with a quite low $\rho=0.25$.

\begin{figure}[t]
\centering
\includegraphics{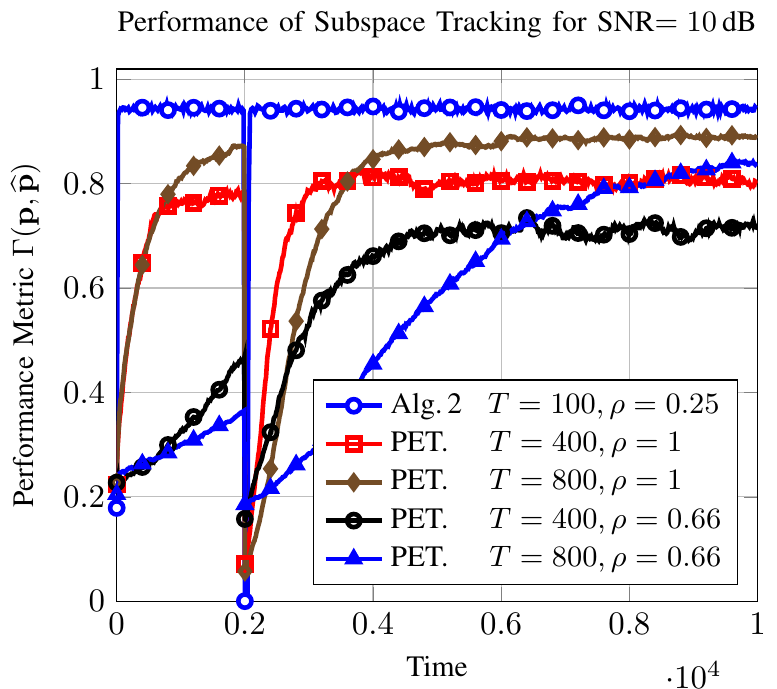}
\caption{{\small Comparison of tracking performance of our algorithm (Alg.\,2) with that of PETRELS for SNR$=10$\,dB.}}
\label{fig:comp_with_petrels}
\end{figure}

\subsection{Simplified Subspace Tracking}
As explained in Section \ref{simple_tracking}, a more low-complexity estimate of the signal subspace at each time $t$ can be obtained by identifying the dominant rows of the weighting matrix $\bfW(t)$. Fig.\,\ref{fig:image_tracking} illustrates the strength of different rows of $\bfW(t)$ at time $t$, corresponding to the estimated received power from different angular grid element in a tracking mode during $t=1, \dots, 400$. We assume that as in the simulations in Section \ref{tracking_perf}, the received angular power distribution of the user undergoes a sharp transition from the angular range $\Theta=[10, 30]$ to $\Theta'=[-40, -20]$ at time $t_\mathsf{tr}=200$. It is seen from Fig.\,\ref{fig:image_tracking} that although there are some spurious rows, our proposed algorithm tracks the location (support) of dominant rows corresponding to the AoAs of the scatterers very well. 
\begin{figure}[t]
\centering
\includegraphics{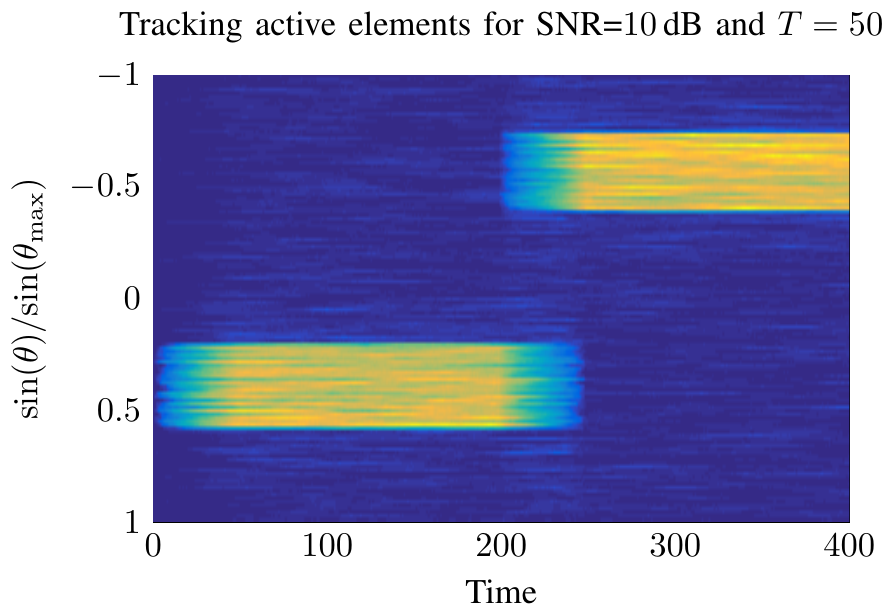}
\caption{An image of the strength of different grid elements estimated by our proposed algorithm during a tracking period $t=1, \dots, 400$. The angular power distribution of the user makes a transition from $\Theta=[10,30]$ to $\Theta'=[-40,-20]$ at time $t_\mathsf{tr}=200$. }
\label{fig:image_tracking}
\end{figure}

\section{Conclusion}
In this paper, we proposed an efficient and low-complexity subspace estimation algorithm, with a special focus towards massive MIMO applications. We mainly studied the AML Algorithm proposed in \cite{haghighatshoar2016channel, haghighatshoar2016massive}, where we showed that the quite slow and time-consuming SDP optimization of AML Algorithm in \cite{haghighatshoar2016channel, haghighatshoar2016massive} (especially when the antenna size $M$ is quite large) can be well approximated with another convex optimization problem, for which we derived a novel iterative low-complexity algorithm. We also considered a generalization of the original AML Algorithm in which the projection (sampling) operator may be time-variant, resulting in a further improvement in subspace estimation. We explained how our  proposed algorithm can be applied to more practical array configurations such as 2D rectangular lattice arrays and provided guidelines for efficient numerical implementation for general array configurations. We also extended our proposed algorithm such that it can be run in the online tracking mode. We evaluated the estimation/tracking performance of our algorithm empirically via numerical simulations and compared it with the performance of other state of the art subspace estimation/tracking algorithms in the literature.

\appendices
\section{Proof of Proposition \ref{convex_eq}}\label{app:convex_eq}
The proof follows by extending Theorem 1 in \cite{steffens2016compact}. The key observation is that for a vector $\bfw \in \bC^T$, the $l_2$-norm $\|\bfw\|$ of $\bfw$ can be written as the output of the following optimization 
\begin{align}\label{eq:norm_rep}
\|\bfw\|=\min_{\bfv \in \bC^T,\,s \in \bC:\, s\bfv =\bfw} \frac{\|\bfv\|^2 + |s|^2}{2}.
\end{align}
In particular, denoting by $(\bfv^*,s^*)$ the optimal solution of \eqref{eq:norm_rep}, we have that $\|\bfw\|=|s^*|^2$. Applying this to the rows of a $G\times T$ coefficient matrix $\bfW$ as in \eqref{eq:l2_1}, we obtain that 
\begin{align}\label{l21_rep}
\|\bfW\|_{2,1}=\min_{\bfV\in \bC^{G\times T},\, \Gammam\in \clD:\, \Gammam \bfV=\bfW} \frac{\|\bfV\|^2 + \|\Gammam\|^2}{2},
\end{align}
where $\clD$ denotes the space of $G\times G$ diagonal matrices with diagonal elements in $\bC$, and where $\Gammam=\diag(\gamma_1, \dots, \gamma_G)\in \clD$. In particular, similar to \eqref{eq:norm_rep} we have that $\|\bfW_{i,.}\|=|\gamma_i^*|^2$, where $\Gammam^*=\diag(\gamma_1^*, \dots, \gamma_G^*)$ is the optimal solution of \eqref{l21_rep}. Replacing $\|\bfW\|_{2,1}$ in \eqref{eq:l2_1} with \eqref{l21_rep}, we can transform \eqref{eq:l2_1} into the following optimization problem 
\begin{align*}
(\bfV^*, \Gammam^*)=\argmin_{\bfV\in \bC^{G\times T},\,\Gammam\in \clD} \frac{\|\check{\bfG} \Gammam \bfV - \bfX\|^2}{\sqrt{T}} + \|\bfV\|^2 + \|\Gammam\|^2.
\end{align*}
For a fixed $\Gammam$, the minimizing matrix $\bfV$ as a function of $\Gammam$ can be obtained via a least-square minimization, where after replacing the solution and applying the \textit{matrix inversion lemma} \cite{hager1989updating} and further simplifications, we obtain the following optimization in terms of $\Gammam$
\begin{align*}
\Gammam^*=\argmin_{\Gammam \in \clD} \tr\Big ( (\check{\bfG} \frac{ \Gammam \Gammam^\herm }{\sqrt{T} }\check{\bfG}^\herm + \bfI_m)^{-1} \widehat{\bfC}_x \Big ) + \tr(\frac{\Gammam \Gammam^\herm}{\sqrt{T}}).
\end{align*}
This optimization can be reparameterized with $\bfP=\frac{\Gammam \Gammam^\herm}{\sqrt{T}}=\diag(p_1, \dots, p_G) \in \clD_+$, where $p_i=\frac{|\gamma_i|^2}{\sqrt{T}} \in \bR_+$, for $i\in [G]$, and where $\clD_+$ denotes the set of all $G\times G$ diagonal matrices with positive diagonal elements. With this parametrization, we obtain the following optimization for the matrix $\bfP \in \clD_+$:
\begin{align}\label{l21_rep3}
\bfP^*=\argmin_{\bfP \in \clD_+} \tr\Big ( (\check{\bfG} \bfP \check{\bfG}^\herm + \bfI_m)^{-1} \widehat{\bfC}_x \Big ) + \tr(\bfP),
\end{align}
Moreover, denoting by $\bfP^*=\diag(p_1^*, \dots, p_G^*)$ the solution of \eqref{l21_rep3}, we have the following relation 
\begin{align}\label{good_rel}
p^*_i=\frac{|\gamma^*_i|^2}{\sqrt{T}}=\frac{\|\bfW^*_{i,.}\|}{\sqrt{T}},
\end{align}
 between the optimal solution $\bfW^*$ of \eqref{eq:l2_1} and the optimal solution $\bfP^*=\diag(p^*_1, \dots, p^*_G)$ of \eqref{l21_rep3}.  Note that as in Section \ref{sec:eq_conv}, we assume that the grid $\clG$ is dense enough such that any signal covariance matrix can be well approximated by 
\begin{align}\label{eq:s_rel}
\bfS \approx \frac{1}{m} \bfG \bfP \bfG^\herm=\frac{1}{m} \bfG\, \diag(p_1,\dots, p_G) \bfG^\herm,
\end{align}
 for some appropriate $\bfP \in \clD_+$ with $p_i\geq 0$, $i \in [G]$. This implies that the term $\check{\bfG} \bfP \check{\bfG}^\herm$ in \eqref{l21_rep3} can be replaced with $\bfB \bfS \bfB^\herm$ (recall that $\check{\bfG}=\frac{1}{\sqrt{m}} \bfB \bfG$), where $\bfS$ takes values from the convex set of all feasible signal covariance matrices, which in the case of a ULA coincides with $\clT_+$. Also note that due to 0-1 sampling matrix $\bfB$ and the special structure of the array responses, every column of $\check{\bfG}$ has a unit $l_2$-norm, which implies that 
\begin{align}
\tr(\check{\bfG} \bfP \check{\bfG}^\herm)=\sum_{i=1}^G p_i \|\check{\bfg}_i\|^2=\sum_{i=1}^G p_i=\tr(\bfP),
\end{align}
where $\check{\bfg}_i$ denotes the $i$-th column of $\check{\bfG}$, which has a unit norm. Again replacing $\check{\bfG} \bfP \check{\bfG}^\herm$ by $\bfB \bfS \bfB^\herm$, it results that \eqref{l21_rep3} is well approximated by the following convex optimization
\begin{align}\label{last_eq}
\bfS^*=\argmin_{\bfS \in \clT_+} \tr\Big ( (\bfB \bfS \bfB^\herm + \bfI_m)^{-1} \widehat{\bfC}_x \Big ) + \tr(\bfB \bfS \bfB^\herm).
\end{align}
Using the well-known Schur's complement condition for positive semi-definiteness (see \cite{boyd1994linear} page 28), we can write \eqref{last_eq} in the form of SDP \eqref{eq:aml} for the AML Algorithm as in \cite{haghighatshoar2016channel,haghighatshoar2016massive}. In particular, having the optimal solution $\bfW^*$ of \eqref{eq:l2_1}, or the optimal solution $\bfP^*$ of \eqref{l21_rep3}, from the parametrization in \eqref{eq:s_rel}, the optimal solution $\bfS^*$ of the SDP \eqref{eq:aml} can be approximated by $\bfS^*={\bfG} \,\diag(s^*_1, \dots, s^*_G) {\bfG}^\herm$, where $s^*_i=\frac{p^*_i}{m}=\frac{\|\bfW^*_{i,.}\|}{m\sqrt{T}}$, for $i\in[G]$, and where we used \eqref{good_rel}. This completes the proof.

\section{Proof of Proposition \ref{sorrogate}}\label{app:sorrogate}
Let $s \in [0,1]$ and let us define $\Deltam(s):=\bfW'+s(\bfW-\bfW')$ and $h(s):=f_1(\Deltam(s))$. We have
\begin{align*}
f_1(\bfW)&-f_1(\bfW')=f_1(\Deltam(1))-f_1(\Deltam(0))\\
&=h(1)-h(0)=\int_{0}^1 h'(s) ds\\
&=\int_{0}^1 \inpr{\nabla f_1(\Deltam(s))}{\bfW-\bfW'} ds\\
&=\inpr{\nabla f_1(\bfW')}{\bfW-\bfW'} \\
&+ \int_{0}^1 \inpr{\nabla f_1(\Deltam(s))- \nabla f_1(\bfW')}{\bfW-\bfW'} ds\\
&\stackrel{(a)}{\leq} \inpr{\nabla f_1(\bfW')}{\bfW-\bfW'} \\
&+ \int_0^1  \beta \|\Deltam(s)-\bfW'\| \|\bfW - \bfW'\| ds\\
&= \inpr{\nabla f_1(\bfW')}{\bfW-\bfW'} + \int_0^1 s \beta \|\bfW-\bfW'\|^2ds\\
&=\inpr{\nabla f_1(\bfW')}{\bfW-\bfW'} + \frac{\beta}{2} \|\bfW-\bfW'\|^2,
\end{align*}
where in $(a)$ we used the Cauchy-Schwarz inequality and the Lipschitz property of $\nabla f_1$. This completes the proof.

\balance
\bibliographystyle{IEEEtran}
{\small
\bibliography{references}}

\end{document}